%% file: main.tex
\documentclass[a4paper]{article}

\usepackage{fullpage}

\usepackage[dvipsnames]{xcolor}
\usepackage{todonotes}
\usepackage{xspace}
\usepackage{complexity}
\usepackage{microtype}
\usepackage[linesnumbered,ruled]{algorithm2e}
\usepackage{tikz}
\usepackage{amsmath}
\usepackage{amsthm}
\usepackage{amssymb}
\usepackage{doi}
\usepackage{authblk}
\usepackage{subcaption}
\usepackage{enumitem}

\graphicspath{{./graphics/}}

\newcommand{\pmreconf}{\textsc{perfect matching reconfiguration}\xspace}
\newcommand{\YES}{\ensuremath{\mathsf{yes}}\xspace}
\newcommand{\NO}{\ensuremath{\mathsf{no}}\xspace}

\newcommand{\ini}{\textup{s}}
\newcommand{\tar}{\textup{t}}

\newcommand{\onestep}[1]{\ensuremath{\overset{#1}{\leftrightarrow}}}
\newcommand{\sevstep}[1]{\ensuremath{\overset{#1}{\leftrightsquigarrow}}}

\newcounter{one}
\newcounter{two}
\newcounter{three}
\newtheorem{theorem}{Theorem}
\newtheorem{lemma}[theorem]{Lemma}

\newtheorem{corollary}[theorem]{Corollary}
\theoremstyle{definition}
\theoremstyle{remark}
\newtheorem{claim}{Claim}

\usetikzlibrary{decorations.pathmorphing,fit,backgrounds,calc}
\usetikzlibrary{decorations.pathreplacing} 
\tikzset{
	edge/.style={thick, gray},
	medge/.style={decorate,very thick,decoration={snake}},
	nedge/.style={very thick,dashed,black},
	vertex/.style={shape=circle,thick,draw,node distance=3em,inner sep=0pt,minimum size=0.5em,fill=white},
	terminal/.style={shape=circle,thick,draw,BrickRed,node distance=3em,inner sep=0pt,minimum size=0.75em,fill=white}
}

\newenvironment{listing}[1]{%
	\begin{list}{*}{%
			\settowidth{\labelwidth}{#1}%
			\setlength{\leftmargin}{\labelwidth}%
			\advance \leftmargin by 12pt
			\setlength{\itemsep}{0pt}%
			\setlength{\parsep}{0pt}%
			\setlength{\topsep}{0pt}%
			\setlength{\parskip}{0pt}%
		}%
	}{%
\end{list}}

\newcommand{\AlgOP}{\textbf{PMROG}}

\title{The Perfect Matching Reconfiguration Problem\footnote{Partially supported by JSPS and MAEDI under the Japan-France Integrated Action Program	(SAKURA).}}

\author[1]{Marthe Bonamy}
\author[2]{Nicolas Bousquet}
\author[3]{Marc Heinrich}
\author[4]{Takehiro Ito\thanks{Partially supported by JST CREST Grant Number JPMJCR1402, and JSPS KAKENHI Grant Numbers JP18H04091 and JP19K11814, Japan.}}
\author[5]{Yusuke Kobayashi\thanks{Partially supported by JST ACT-I Grant Number JPMJPR17UB, and JSPS KAKENHI Grant Numbers JP16K16010 and JP18H05291, Japan.}}
\author[6]{Arnaud Mary}
\author[7]{Moritz M\"uhlenthaler}
\author[8]{Kunihiro Wasa\thanks{Partially supported by JST CREST Grant Numbers JPMJCR18K3 and JPMJCR1401, and JSPS KAKENHI Grant Number JP19K20350, Japan.}}

\affil[1]{CNRS, LaBRI, Universit\'e de Bordeaux, Talence, France}
\affil[2]{Univ. Grenoble Alpes, CNRS, G-SCOP, Grenoble-INP, Grenoble, France}
\affil[3]{Universit\'e Claude Bernard Lyon 1, LIRIS, UMR5205, France}
\affil[4]{Graduate School of Information Sciences, Tohoku University, Japan}
\affil[5]{Research Institute for Mathematical Sciences, Kyoto University, Japan}
\affil[6]{LBBE, Universit\'e Claude Bernard Lyon 1, Lyon, France}
\affil[7]{Fakult\"at f\"ur Mathematik, TU Dortmund University, Germany}
\affil[8]{National Institute of Informatics, Japan}

\footnotetext{The first, second, third and sixth authors are partially supported by ANR project GrR (ANR-18-CE40-0032).}

\begin{document}

\maketitle
\begin{abstract}
	We study the {\sc perfect matching reconfiguration} problem: Given two
	perfect matchings of a graph, is there a sequence of flip operations
	that transforms one into the other?  Here, a flip operation exchanges
	the edges in an alternating cycle of length four.  We are interested in
	the complexity of this decision problem from the viewpoint of graph
	classes.  We first prove that the problem is PSPACE-complete even for
	split graphs and for bipartite graphs of bounded bandwidth with
	maximum degree five.  We then investigate polynomial-time solvable
	cases.  Specifically, we prove that the problem is solvable in
	polynomial time for strongly orderable graphs (that include interval
	graphs and strongly chordal graphs), for outerplanar graphs, and for
	cographs (also known as $P_4$-free graphs).  Furthermore, for each
	yes-instance from these graph classes, we show that a linear number of flip
	operations is sufficient and  we can exhibit a corresponding sequence of
	flip operations in polynomial time.
%
\end{abstract}
\newpage

\input{introduction.tex}

\input{preliminary.tex}
\input{hardness.tex}
\section{Polynomial-time algorithms} \label{sec:polytime}
In this section, we investigate the polynomial-time solvability of {\sc perfect matching reconfiguration} from the viewpoint of graph classes. 
	\input{interval.tex}
	\input{outerplanar.tex}

	\input{cograph.tex}

\input{conclusion.tex}

\newpage
\providecommand{\noopsort}[1]{}

\newpage
\appendix
\input{appendix.tex}




\end{document}

%% file: introduction.tex
\section{Introduction}
\label{sec:introduction}
Given an instance of some combinatorial search problem and two of its
feasible solutions, a \emph{reconfiguration problem} asks whether one solution
can be transformed into the other in a step-by-step fashion, such that each
intermediate solution is also feasible.
Reconfiguration problems capture dynamic situations, where some
solution is in place and we would like to move to a desired alternative
solution without becoming infeasible. A systematic study of the complexity of
reconfiguration problems was initiated in~\cite{Ito:11}.
Recently the topic has gained a lot of attention in the context of constraint
satisfaction problems and graph problems, such as the independent set problem,
the matching problem, and the dominating set problem. Reconfiguration problems
naturally arise for operational research problems but also are closely related
to uniform sampling (using Markov chains) or enumeration of solutions of a
problem.  For an overview of recent results on reconfiguration problems, the
reader is referred to the surveys of van den Heuvel~\cite{vHeuvel13} and
Nishimura~\cite{Nishimura17}.

In order to define valid step-by-step transformations, an adjacency relation on
the set of feasible solutions is needed.
%
Depending on the problem, there may be different natural choices of adjacency
relations. For instance, we may assume that two matchings of a graph are
adjacent if one can be obtained from the other by exchanging precisely one
edge, i.e., there exists $e \in M$ and $f \in M'$ such that $M
\setminus \{ e \} = M^\prime \setminus \{ f \}$.  
The corresponding modification of a matching is usually referred to as \emph{token jumping}
(TJ). Here, the  \emph{tokens} are the edges of a matching and a token may be
``moved'' from an edge of the matching to another edge 
so that we obtain the another matching.
On can similarly another adjacency relation, where two matchings are adjacent
if one can be obtained from the other by moving a token to some incident edge.
The adjacency relation is called \emph{token sliding}
(TS).  Ito et al.~\cite{Ito:11} gave a polynomial-time algorithm that decides
if there is a transformation between two given matchings under the TJ and TS
operations.  

\subsection{The Perfect Matching Reconfiguration Problem}
\label{subsec:problem}

Recall that a matching of a graph is \emph{perfect} if it covers each
vertex. We study the complexity of deciding if there is a step-by-step
transformation between two given perfect matchings of a graph.
However, according to the adjacency relations given by the TS and TJ
operations, there is no transformation between any two distinct perfect
matchings of a graph. Since the symmetric difference of any two perfect
matchings of a graph consists of even-length disjoint cycles, it is natural to
consider a different adjacency relation for perfect matchings. We say that two
perfect matchings of a graph differ by a \emph{flip} (or \emph{swap}) if
their symmetric difference induces a cycle of length four. We consider two
perfect matchings to be are adjacent if they differ by a flip.  Intuitively,
for two adjacent perfect matchings $M$ and $M'$, we think of a flip as an
operation that exchanges edges in $M \setminus M'$ for edges in $M' \setminus
M$. A flip is in some sense a minimal modification of a perfect matching.

\begin{figure}[tb]
    \begin{center}
        \includegraphics[width=0.75\linewidth]{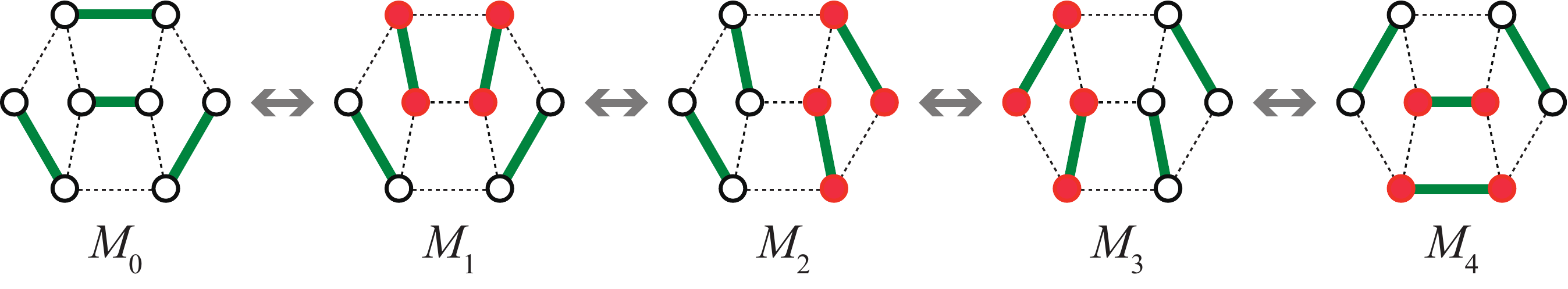}
    \end{center}
    \vspace{-1em}
    \caption{A transformation between perfect matchings $M_0$ and $M_4$ under
        the flip operation. For $1 \leq i \leq 4$, the matching $M_{i}$ can be
        obtained from $M_{i-1}$ by applying the flip operation to the cycle
        induced by the four painted (red) vertices in $M_{i}$.}
    \label{fig:flip}
\end{figure}

An example of a transformation between two perfect matchings of a graph is
given in Figure~\ref{fig:flip}.
We formalize the task of deciding the existence of transformation between two
given perfect matchings as follows.
\bigskip

\begin{quote}
    \textsc{Perfect Matching Reconfiguration}\\
    \textbf{input:} Graph $G$, perfect matchings $M_{\ini}$ and $M_{\tar}$ of $G$.\\
    \textbf{question:} Is there a sequence of flips that transforms $M_{\ini}$ into $M_{\tar}$?
\end{quote}

Note that if we do not restrict the length of a cycle in the definition of a
flip, then for any two perfect matchings $M_{\ini}$ and $M_{\tar}$ of a graph, there is a
sequence of flip operations that transforms $M_{\ini}$ into $M_\tar$, since we can perform
a flip on each cycle of the symmetric difference of $M_\ini$ and $M_\tar$.
As a compromise, we may extend the problem definition to flips on cycles of
fixed constant length $k$, where $k > 4$ and $k$ is even. We refer to the
corresponding reconfiguration problem as $k$-\textsc{Perfect Matching
Reconfiguration}.

\subsection{Related Work}

Transformations between matchings have been studied in various settings. Transformation 
of matchings using flips has been considered for generating random matchings.
Numerous algorithms and hardness results are available for finding
transformations between matchings ---  and more generally, independent sets ---
using the TS and TJ operations.
Furthermore, the flip operation is well-known for stable matchings and some
geometric matching problems related to finding transformations between
triangulations.

\paragraph{Sampling Random Matchings}

The problem of sampling or enumerating perfect matchings in a graph received a considerable attention (see e.g.~\cite{StefankovicVW18}). 
Determining the connectivity, and the diameter of the solution space formed by perfect matchings under the flip operation provide some information on the ergodicity or the mixing time of the underlying Markov chain. Indeed, the connectivity of the chain ensures the irreducibility (and usually the ergodicity) of the underlying Markov chain. Additionally, the diameter of the reconfiguration graph provides a lower bound on the mixing time of the chain.

The use of flips for sampling random perfect matchings was first started in~\cite{DGH01} where it is seen as a generalisation of transpositions for permutations. Their work was later improved and generalized in~\cite{DJM17} and~\cite{DM17}. The focus of these last two articles is to investigate the problem of sampling random perfect matchings using a Markov Chain called the switch chain. Starting from an arbitrary perfect matching, the chain proceeds by applying at each step a random flip (called switch in these papers). The aim of these papers is to characterize classes of graphs for which simulating this chain for a polynomial number of steps is enough to generate a perfect matching close to uniformly distributed. Some of their results can be reformulated in the reconfiguration terminology. In~\cite{DJM17}, it is proved that the largest hereditary class of bipartite graphs for which the reconfiguration graph of perfect matchings with flips is connected is the class of chordal bipartite graphs. This result is generalized in~\cite{DM17} where they characterize the hereditary class of general (non-bipartite) graphs for which the reconfiguration graph is connected. They call this class \textsc{Switchable}. Note that it is not clear whether graphs in this class can be recognized in polynomial time. The question of the complexity of \pmreconf is also mentioned in~\cite{DM17}.

\paragraph{Reconfiguration of Matchings and Independent Sets.}
Recall that matchings of a graph correspond to independent sets of its line
graph.  Although reconfiguration of independent sets received a considerable
attention in the last decade
(e.g.,~\cite{BonsmaKW14,BousquetMP17,DemaineDFHIOOUY14,HearnD05,ItoKOSUY14,KaminskiMM12,Wrochna18}),
all the known results for reconfiguration of independent sets are based on the
TJ or TS operations as adjacency relations.  Thus, none of these results carry over
to the {\sc Perfect Matching Reconfiguration} problem.

A related problem can be found in a more general setting: 
The problem of determining, enumerating, or randomly generating graphs with a fixed degree sequence has received a considerable attention since the fifties (see e.g.~\cite{Senior51,HakiII,Will99}). Given two graphs with a fixed degree sequence, one might want to know if it is possible to transform the one into the other via a sequence of flip operations and if yes, how many steps are needed for such a transformation;
note that the host graph (i.e., the graph $G$ in our problem) is a clique in this setting. 
Hakimi~\cite{HakiII} proved that such a transformation always exists. 
Will~\cite{Will99} proved that the problem of finding a shortest transformation
is \NP-complete, and Bereg and Ito~\cite{Bereg17} provide a $\frac
32$-approximation algorithm for this problem. 

	\begin{figure}[tb]
	\begin{center}
		\includegraphics[width=0.9\linewidth]{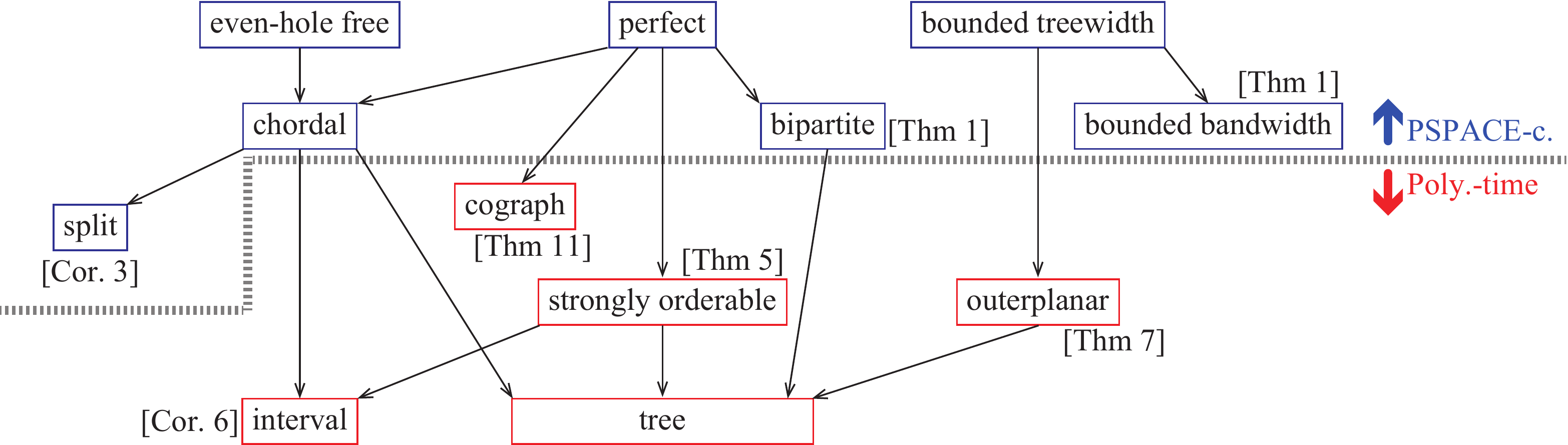}
	\end{center}
\vspace{-1em}
	\caption{Our results, where each arrow represents the inclusion relationship between graph classes: 
		$A \to B$ represents that the graph class $B$ is properly included in the graph class $A$.}
	\label{fig:results}
\end{figure}

\paragraph{Stable Matchings.}
Suppose we are given a bipartite graph and for each vertex a linear preference
order of its neighbors. A matching $M$ is not stable if there is an edge $uv$
not in $M$, such that $u$ prefers $v$ and $v$ prefers $u$ to their respective
$M$-partners. 
The classical algorithm by Gale and Shapley~\cite{GI:89} yields a stable
matching in polynomial time. It is known that any two stable matchings cover
the same vertices, so the stable matchings are perfect matchings of some
subgraph.  Furthermore, they form a distributive lattice under \emph{rotations}
(another word for flips) on preference-orienced cycles, see for
example~\cite{GI:89}.  Essentially, the symmetric difference of two stable
matchings consists of disjoint cycles and we may flip edges on these cycles to
obtain another stable matching. If we drop the preferences, then the question
is simply if we can find a transformation between two perfect matchings by
flipping edges on cycles in the symmetric difference. Clearly the answer is yes,
for example by processing the cycles in the symmetric difference one-by-one. 
We consider a similar setting, but restrict
the length of the cycles.

\paragraph{Flips of Triangulations.} A flip of a triangulation is similar to
flipping an alternating cycle in the sense that we switch between two states of
a quadrilateral. In the context of triangulations, a flip operation
switches the diagonal of a quadrilateral. Transformations between
triangulations of point sets and
polygons using flips have been studied mostly in the plane. It is known that
the flip graph of triangulations of point sets and polygons in the plane is
connected and has diameter $O(n^2)$, where $n$ is the
number of points~\cite{HNU:99,Lawson:77}. Recently, \NP-completeness has been
proved for deciding the flip-distance between triangulations of a point
set in the plane~\cite{LP:15} and triangulations of a simple
polygon~\cite{AMP:15}.

Houle et al.~have considered triangulations of point sets in the plane that
admit a perfect matching~\cite{HHNR:05}. They show that any two such
triangulations are connected under the flip operation. For this purpose they
consider the graph of non-crossing perfect matchings, where two matchings are
adjacent if they differ by a single non-crossing cycle (of arbitrary length).
They show that the graph of non-crossing perfect matchings is connected and
conclude from this that any two triangulations that admit a perfect matching
must be connected.  In contrast to their setting, we remove all geometric
requirements, but restrict the length of the cycles allowed for the flip
operation.

\subsection{Our results}
In this paper, we study the complexity of {\sc Perfect Matching
Reconfiguration} from the viewpoint of graph classes.
\figurename~\ref{fig:results} summarizes our results.

Recall that reconfiguration of matchings under the TS and TJ operations can be
solved in polynomial time for any graph~\cite{Ito:11}.  In contrast, we prove
that {\sc Perfect Matching Reconfiguration} is \PSPACE-complete, even for split
graphs, and for bipartite graphs of bounded bandwidth and of maximum degree
five. 
We extend our hardness result to a more general setting,
namely the reconfiguration of $k$-factor subgraphs. 
Furthermore, by adjusting our gadgets appropriately, we show that $k$-\textsc{Perfect Matching Reconfiguration} is
\PSPACE-complete for any even $k \geq 4$. 
Note that this result contrasts with \textsc{TJ-Matching Reconfiguration} problem that can be decided in polynomial time~\cite{Ito:11} and with geometric reconfiguration problems where the reconfiguration operation is a flip (e.g., flips of triangulations~\cite{LP:15}) which usually belong to NP. 


We additionally investigate polynomial-time solvable cases.  We prove that {\sc
Perfect Matching Reconfiguration} admits a polynomial-time algorithm on
strongly orderable graphs (these include interval graphs and strongly chordal
graphs), outerplanar graphs, and cographs (also known as $P_4$-free graphs).
More specifically, we give the following results: 
\begin{itemize}
	\item For strongly orderable graphs, a transformation between two
	  perfect matchings always exists; hence the answer is always \YES.
	  Furthermore, there is a transformation of linear length (i.e., a
	  linear number of flip operations) between any two matchings and
	  such a transformation can be found in polynomial time. 
				
	\item {\sc Perfect Matching Reconfiguration} on outerplanar graphs can
	  be solved in linear time, and we can find a transformation of linear
	  length for a \YES-instance in linear time.  (Note that there are
	  \NO-instance, e.g., long cycles).
	
	\item {\sc Perfect Matching Reconfiguration} on cographs can be solved
	  in polynomial time, and we can find a transformation of linear length
	  for a \YES-instance in polynomial time. (Again, there are \NO-instances).
\end{itemize}

	Proofs of the claims marked with $(\ast)$ and one figure have been moved to Appendices.

%% file: preliminary.tex
\subsection{Notation}

For standard definitions and notations on graphs, we refer the reader
to~\cite{Diestel}. Let $G = (V, E)$ be a simple graph. Two edges are
\emph{independent} if they share no endpoint. A \emph{matching} $M \subseteq E$
of $G$ is a set of pairwise independent edges.  For a vertex set $V^\prime
\subseteq V$, we denote by $G[V^\prime]$ the subgraph of $G$ induced by
$V^\prime$.  For a vertex $v \in V$, we denote by $N(v)$ the
\emph{neighborhood} of $v$, that is, $N(v) := \{w \in V \mid vw \in E\}$. 

Two matchings $M$ and $M^\prime$ of $G$ are \emph{adjacent} if their symmetric
difference $M \vartriangle M^\prime$ induces a cycle of length four. 
We write $M \onestep{G} M^\prime$ if $M$ and $M^\prime$ are adjacent; we 
may omit $G$ if no confusion is possible.
A sequence $M_0, M_1, \ldots, M_q$ of matchings in $G$ is called a
\emph{reconfiguration sequence between $M$ and $M^\prime$} if  $M_0 = M$, $M_q
= M^\prime$, and for $1 \leq i \leq q$, we have $M_{i-1} \onestep{G} M_i$.
We write $M \sevstep{G} M^\prime$ (or simply $M \sevstep{} M^\prime$) if there
is a reconfiguration sequence between $M$ and $M^\prime$. 
The \textsc{Matching Reconfiguration} problem under the flip operation is defined as follows: 
\begin{center}
  \parbox{0.95\hsize}{
    \begin{listing}{{\bf Question:}}
    \item[{\bf Input:}] A simple graph $G$, and two matchings $M_{\ini}$ and $M_{\tar}$ of $G$
    \item[{\bf Question:}] Determine whether $M_{\ini} \sevstep{G} M_{\tar}$ or not.
    \end{listing}}
  \end{center}

%% file: hardness.tex
\section{PSPACE-completeness} \label{sec:hard}

	In this section, we prove that \textsc{perfect matching reconfiguration} is PSPACE-complete. 
	Interestingly, the problem remains intractable even for bipartite graphs, even though matchings in bipartite graphs satisfy several nice properties. 	
	\begin{theorem} \label{the:bipartite}
		\textsc{Perfect matching reconfiguration} is PSPACE-complete for bipartite graphs whose maximum degree is five and whose bandwidth is bounded by a fixed constant. 
	\end{theorem}
	\begin{proof}
	Observe that the problem can be solved in (most conveniently, nondeterministic~\cite{DBLP:journals/jcss/Savitch70}) polynomial space, and hence it is in PSPACE. 
	As a proof of Theorem~\ref{the:bipartite}, we thus prove that the problem is PSPACE-hard for such graphs, by giving a polynomial-time reduction from the \textsc{Nondeterministic Constraint Logic} problem (\textsc{NCL} for short)~\cite{HearnD05}. 

	\paragraph*{Definition of nondeterministic constraint logic.}

	\begin{figure}[b]
	\begin{center}
		\includegraphics[width=0.75\linewidth]{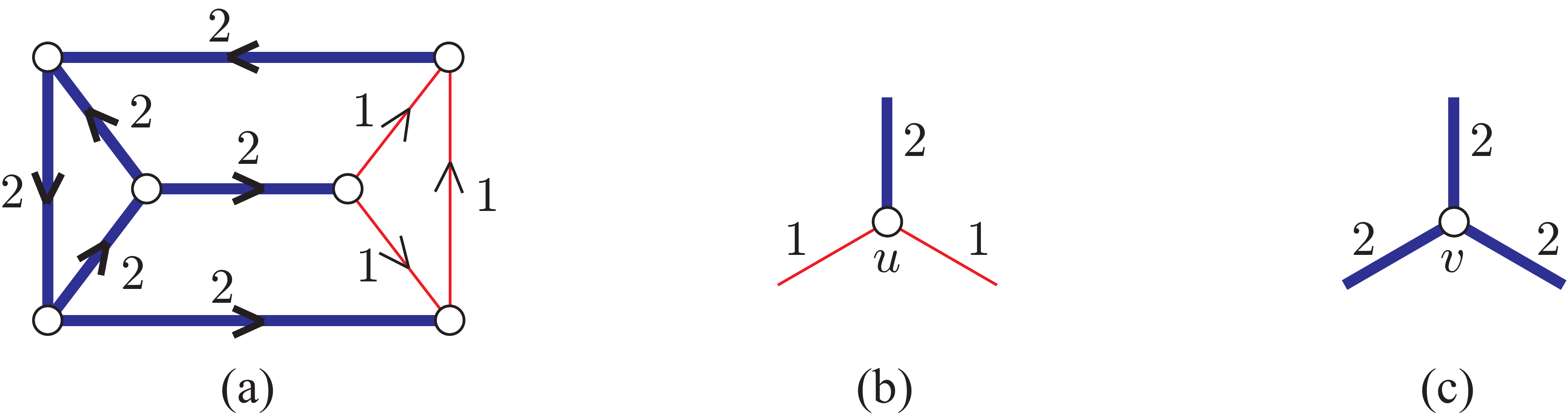}
	\end{center}
	\caption{(a) A configuration of an NCL machine, (b) an NCL \textsc{and} vertex $u$, and (c) an NCL \textsc{or} vertex $v$.}
	\label{fig:ncl}
	\end{figure}

	An NCL ``machine'' is an undirected graph together with an assignment of weights from $\{1,2\}$ to each edge of the graph. 
	An (\emph{NCL}) \emph{configuration} of this machine is an orientation (direction) of the edges such that the sum of weights of in-coming arcs at each vertex is at least two. 
	\figurename~\ref{fig:ncl}(a) illustrates a configuration of an NCL machine, where each weight-$2$ edge is depicted by a (blue) thick line and each weight-$1$ edge by a (red) thin line. 
	Then, two NCL configurations are \emph{adjacent} if they differ in a single edge direction. 
	Given an NCL machine and its two configurations, it is known to be PSPACE-complete to determine whether there exists a sequence of adjacent NCL configurations which transforms one into the other~\cite{HearnD05}. 

	An NCL machine is called an \textsc{and}/\textsc{or} \emph{constraint graph} if it consists of only two types of vertices, called ``NCL \textsc{and} vertices'' and ``NCL \textsc{or} vertices'' defined as follows:
	A vertex of degree three is called an \emph{NCL \textsc{and} vertex} if its three incident edges have weights $1$, $1$, and $2$. 
	(See \figurename~\ref{fig:ncl}(b).)
	An NCL \textsc{and} vertex $u$ behaves as a logical \textsc{and}, in the following sense: 
	the weight-$2$ edge can be directed outward for $u$ only if both two weight-$1$ edges are directed inward for $u$. 
	Note that, however, the weight-$2$ edge is not necessarily directed outward even when both weight-$1$ edges are directed inward. 
	A vertex of degree three is called an \emph{NCL \textsc{or} vertex} if its three incident edges have weights $2$, $2$, and $2$. 
	(See \figurename~\ref{fig:ncl}(c).)
	An NCL \textsc{or} vertex $v$ behaves as a logical \textsc{or}: 
	one of the three edges can be directed outward for $v$ if and only if at least one of the other two edges is directed inward for $v$. 
It should be noted that, although it is natural to think of NCL \textsc{and}/\textsc{or} vertices as having inputs and outputs, there is nothing enforcing this interpretation; 
especially for NCL \textsc{or} vertices, the choice of input and output is entirely arbitrary because an NCL \textsc{or} vertex is symmetric. 
For example, the NCL machine in \figurename~\ref{fig:ncl}(a) is an \textsc{and}/\textsc{or} constraint graph. 
From now on, we call an \textsc{and}/\textsc{or} constraint graph simply an \emph{NCL machine}, and call an edge in an NCL machine an \emph{NCL edge}. 
NCL remains PSPACE-complete even if an input NCL machine is planar, bounded bandwidth, and of maximum degree three~\cite{DBLP:conf/iwpec/Zanden15}.

	\paragraph*{Gadgets.} 
	Suppose that we are given an instance of \textsc{NCL}, that is, an NCL machine and two configurations of the machine. 
	We will replace each of NCL edges and NCL {\sc and}/{\sc or} vertices with its corresponding gadget;
if an NCL edge $e$ is incident to an NCL vertex $v$, then we connect the corresponding gadgets for $e$ and $v$ by a pair of vertices, called {\em connectors} ({\em between $v$ and $e$}) or {\em $(v,e)$-connectors}, as illustrated in \figurename~\ref{fig:subdivision}(a) and (b).
	Thus, each edge gadget has two pairs of connectors, and each {\sc and}/{\sc or} gadget has three pairs of connectors.  
	Our gadgets are all edge-disjoint, and share only connectors. 

	In our reduction, we construct the correspondence between orientations of an NCL machine and perfect matchings of the corresponding graph, as follows:
	We regard that the orientation of an NCL edge $e=vw$ is inward direction for $v$ if the two $(v,e)$-connectors are both covered by (edges in) the {\sc and}/{\sc or} gadget for $v$. 
	On the other hand, we regard that the orientation of $e=vw$ is outward direction for $w$ if the two $(w,e)$-connectors are both covered by the edge gadget for $e$.
%
	\figurename~\ref{fig:gadget} shows our three types of gadgets which correspond to NCL edges and NCL {\sc and}/{\sc or} vertices. 
	We below explain the behavior of each gadget.  

	\begin{figure}[tb]
	\begin{center}
		\includegraphics[width=0.65\linewidth]{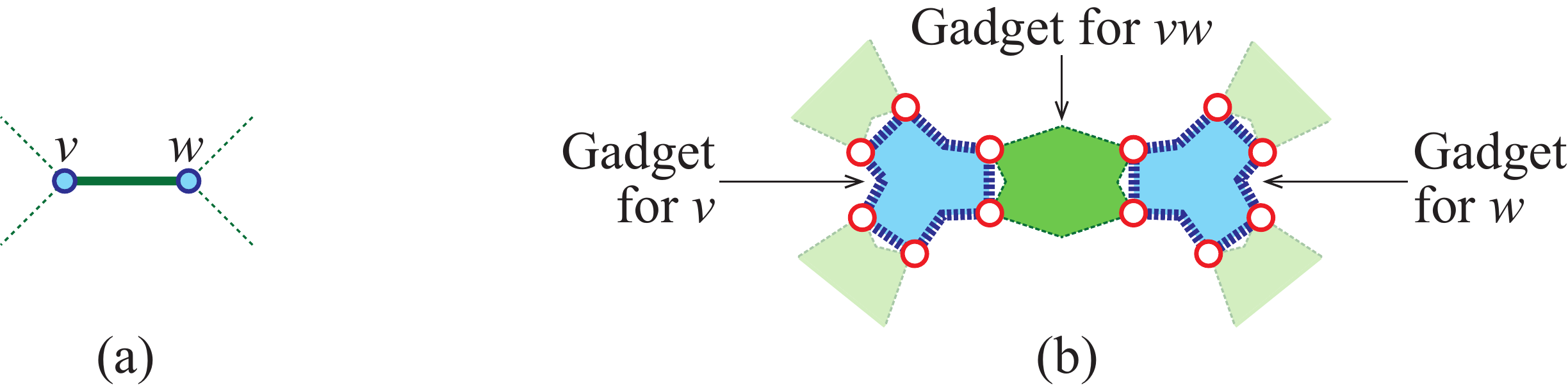}
	\end{center}
	\caption{(a) An NCL edge $vw$, and (b) its corresponding gadgets, where the connectors are depicted by (red) circles.}
	\label{fig:subdivision}
\end{figure}

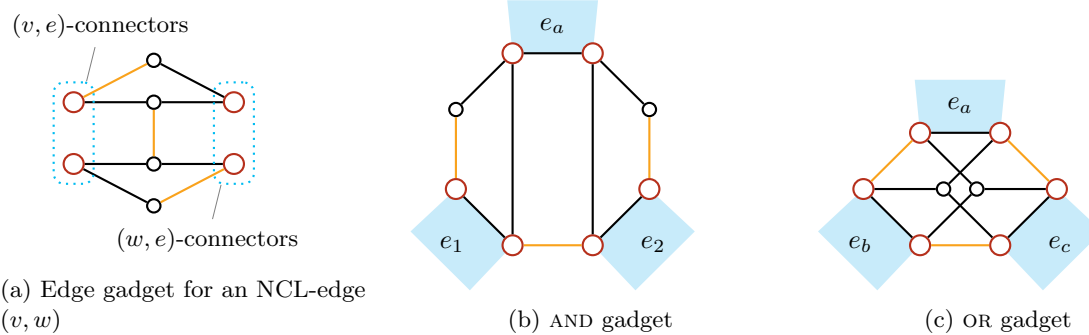
\begin{figure}[tb]
	\begin{center}
	  \begin{subfigure}[b]{0.3\linewidth}
	    \begin{tikzpicture}[node/.append style={node distance=1.5em}]
	      \node[terminal] (lv1) {};
	      \node[terminal,below=1.5em of lv1] (lv2) {};
	      \node[vertex,right of=lv1] (mv2) {};
	      \node[vertex,right of=lv2] (mv3) {};
	      \node[vertex,above=1em of mv2] (mv1) {};
	      \node[vertex,below=1em of mv3] (mv4) {};
	      \node[terminal,right of=mv2] (rv1) {};
	      \node[terminal,right of=mv3] (rv2) {};

	      \draw[thick,YellowOrange] (lv1) -- (mv1);
	      \draw[thick] (mv1) -- (rv1); 
	      \draw[thick,YellowOrange] (rv2) -- (mv4);
	      \draw[thick] (mv4) -- (lv2);
	      \draw[thick] (lv1) -- (mv2) -- (rv1);
	      \draw[thick] (lv2) -- (mv3) -- (rv2);
	      \draw[thick,YellowOrange] (mv2) -- (mv3);

	      \node[draw, dotted, thick, ProcessBlue!80, fit=(lv1) (lv2), rounded corners, pin={[xshift=1em]above:{\small $(v,e)$-connectors}}] {};
	      \node[draw, dotted, thick, ProcessBlue!80, fit=(rv1) (rv2), rounded corners, pin={[xshift=-1em]below:{\small $(w,e)$-connectors}}] {};
	    \end{tikzpicture}
	    \subcaption{Edge gadget for an NCL-edge $(v, w)$}
	  \end{subfigure}
	  \hspace{1em}
	  \begin{subfigure}[b]{0.3\linewidth}
	    \begin{tikzpicture}[node/.append style={node distance=1.5em},fitbox/.style={fill=blue!15, rounded corners,minimum size=1.5em}]
	      \node[terminal] (ta1) {};
	      \node[terminal,right of=ta1] (ta2) {};
	      \node[vertex,below right of=ta2] (rv1) {};
	      \node[vertex,below left of=ta1] (lv1) {};
	      \node[terminal,below of=lv1] (tb1) {};
	      \node[terminal,below right of=tb1] (tb2) {};
	      \node[terminal,below of=rv1] (tc1) {};
	      \node[terminal,below left of=tc1] (tc2) {};

	      \node[draw=none,above=1em of ta1] (da) {};

	      \draw[thick] (ta1) -- node[label=above:$e_a$] {} (ta2) -- (rv1);
	      \draw[thick,YellowOrange] (rv1) -- (tc1);
	      \draw[thick] (tc1) -- node[label=below right:$e_2$] {} (tc2);
	      \draw[thick,YellowOrange] (tc2) -- (tb2);
	      \draw[thick] (tb2) -- node[label=below left:$e_1$] {} (tb1);
	      \draw[thick,YellowOrange] (tb1) -- (lv1);
	      \draw[thick] (lv1) -- (ta1);
	      \draw[thick] (tc2) --  (ta2);
	      \draw[thick] (tb2) -- (ta1);

	      \begin{scope}[on background layer]
		\draw[draw=none,fill=ProcessBlue!20] let 
		    \p1 = (ta1.center),
		    \p2 = (ta2.center),
		    \p3 = (0.2em,2em) in
		    (\x1-\x3,\y1+\y3) -- (\p1) -- (\p2) -- (\x2+\x3,\y2+\y3);

		\draw[draw=none,fill=ProcessBlue!20] let 
		    \p1 = (tb1.center),
		    \p2 = (tb2.center) in
		    (\x1-1.7em,\y1-1.6em) -- (\p1) -- (\p2) -- (\x2-1.6em,\y2-1.7em);

		\draw[draw=none,fill=ProcessBlue!20] let 
		    \p1 = (tc1.center),
		    \p2 = (tc2.center) in
		    (\x2+1.6em,\y2-1.7em) -- (\p2) -- (\p1) -- (\x1+1.7em,\y1-1.6em);
	      \end{scope}
	      
%
%
	    \end{tikzpicture}
	    \subcaption{\textsc{and} gadget}
	  \end{subfigure}
	  \hspace{1em}
	  \begin{subfigure}[b]{0.3\linewidth}
	    \begin{tikzpicture}[node/.append style={node distance=1.5em},fitbox/.style={fill=blue!15, rounded corners,minimum size=1.5em}]
	      \node[terminal] (ta1) {};
	      \node[terminal,right of=ta1] (ta2) {};
	      \node[terminal,below right of=ta2] (tc1) {};
	      \node[terminal,below left of=tc1] (tc2) {};
	      \node[terminal,below left of=ta1] (tb1) {};
	      \node[terminal,below right of=tb1] (tb2) {};

	      \node[vertex,below right of=ta1] (c1) {};
	      \node[vertex,below left of=ta2] (c2) {};

	      \draw[thick] (ta1) -- node[label=above:$e_a$] {} (ta2);
	      \draw[thick,YellowOrange] (ta2) -- (tc1);
	      \draw[thick] (tc1) -- node[label=below right:$e_c$] {} (tc2);
	      \draw[thick,YellowOrange] (tc2) -- (tb2);
	      \draw[thick] (tb2) -- node[label=below left:\textcolor{black}{$e_b$}] {} (tb1);
	      \draw[thick,YellowOrange] (tb1) -- (ta1);
	      \draw[thick] (ta1) -- (c1) -- (tb2);
	      \draw[thick] (ta2) -- (c2);
	      \draw[thick] (c2) -- (tc2);
	      \draw[thick] (tb1) -- (c2);
	      \draw[thick] (tc1) -- (c1);

	      \begin{scope}[on background layer]
		\draw[draw=none,fill=ProcessBlue!20] let 
		    \p1 = (ta1.center),
		    \p2 = (ta2.center),
		    \p3 = (0.2em,2em) in
		    (\x1-\x3,\y1+\y3) -- (\p1) -- (\p2) -- (\x2+\x3,\y2+\y3);

		\draw[draw=none,fill=ProcessBlue!20] let 
		    \p1 = (tb1.center),
		    \p2 = (tb2.center) in
		    (\x1-1.7em,\y1-1.6em) -- (\p1) -- (\p2) -- (\x2-1.6em,\y2-1.7em);

		\draw[draw=none,fill=ProcessBlue!20] let 
		    \p1 = (tc1.center),
		    \p2 = (tc2.center) in
		    (\x2+1.6em,\y2-1.7em) -- (\p2) -- (\p1) -- (\x1+1.7em,\y1-1.6em);
	      \end{scope}
	    \end{tikzpicture}
	    \subcaption{\textsc{or} gadget}
	  \end{subfigure}
	\end{center}
	\caption{Illustrations of the three gadgets. In the {\sc and}/{\sc or}
	gadget, the three light blue parts represent the edge gadgets
	corresponding to the edges incident to the NCL vertex; $e_1$ and
	$e_2$ in the {\sc and} gadget correspond to weight-$1$ edges. The
	orange edges may be subdivided in order to make the gadgets work for
        swaps on longer cycles.}
	\label{fig:gadget}
\end{figure}

	\begin{figure}[tb]
	\begin{center}
		\includegraphics[width=\linewidth]{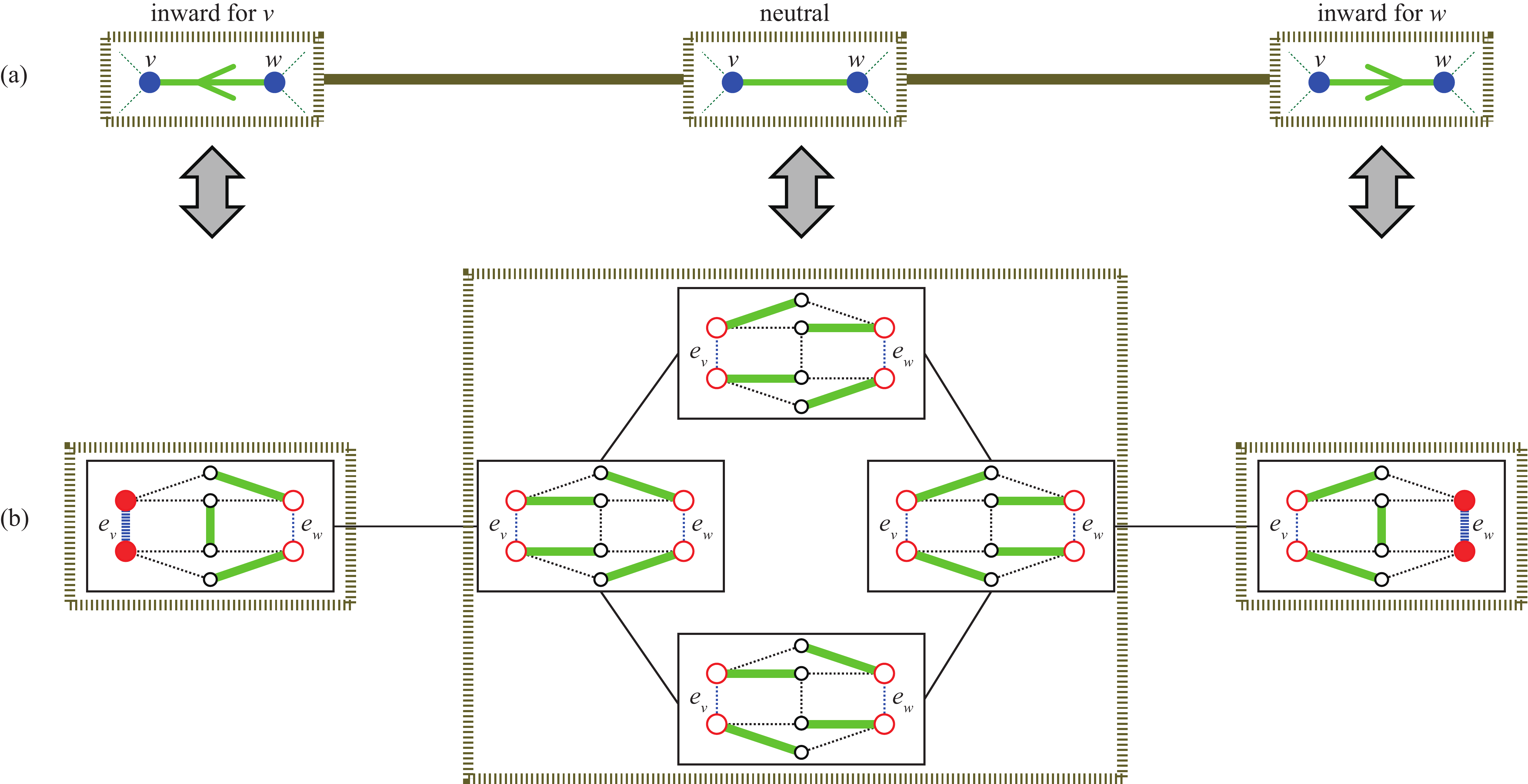}
	\end{center}
	\caption{(a) A reversal of the orientation of an NCL edge $vw$ via its neutral orientation, and (b)~all configurations of the edge gadget. Note that two edges $e_v$ and $e_w$ joining connectors do not belong to the edge gadget, but to the {\sc and}/{\sc or} gadgets for $v$ and $w$, respectively. The inside of each connector is painted (by red) if it is covered by the {\sc and}/{\sc or} gadget for $v$ or $w$.}
	\label{fig:edge_configuration}
\end{figure}

\medskip

\noindent
	{\bf Edge gadget.}
	Recall that, in a given NCL machine, two incident NCL vertices $v$ and $w$ are joined by a single NCL edge $e=vw$. 
	Therefore, the edge gadget for $vw$ should be consistent with the orientations of the NCL edge $vw$, as follows (see also \figurename~\ref{fig:edge_configuration}):
If $(v,e)$-connectors are both covered by the {\sc and}/{\sc or} gadget for $v$ (i.e., the inward direction for $v$), then $(w,e)$-connectors must be covered by the edge gadget for $e$ (i.e., the outward direction for $w$); 
conversely, the $(v,e)$-connectors must be covered by the edge gadget for $e$ if $(w,e)$-connectors are covered by the {\sc and}/{\sc or} gadget for $w$. 
	In particular, the edge gadget must forbid a configuration such that all $(v,e)$- and $(w,e)$-connectors are covered by the {\sc and}/{\sc or} gadgets for $v$ and $w$, respectively (i.e., the inward directions for both $v$ and $w$), because such a configuration corresponds to the direction which illegally contributes to both $v$ and $w$ at the same time. 
	On the other hand, covering all $(v,e)$- and $(w,e)$-connectors by the edge gadget at the same time (i.e., the outward directions for both $v$ and $w$) corresponds to the {\em neutral} orientation~\cite{DBLP:journals/ieicet/OsawaSIZ18} of the NCL edge $e=vw$ which contributes to neither $v$ nor $w$, and hence we simply do not care about such orientations. 

	\figurename~\ref{fig:gadget}(a) illustrates our edge gadget for an NCL edge $e = vw$.
	Then, if all $(v,e)$- and $(w,e)$-connectors are covered by the {\sc and}/{\sc or} gadgets for $v$ and $w$, then no matching can cover all four vertices in the middle (in particular, we cannot cover the top and bottom vertices of the four).  
	Thus, this edge gadget forbids the orientation of $e$ which gives the inward directions for both $v$ and $w$ at the same time. 

	\figurename~\ref{fig:edge_configuration}(b) illustrates valid configurations of the edge gadget for $e=vw$ together with two edges $e_v$ and $e_w$ from the {\sc and}/{\sc or} gadgets for $v$ and $w$, respectively. 
	Each (non-dotted) box represents a valid configuration, and two boxes are joined by an edge if their configurations are adjacent, that is, can be obtained by flipping a single cycle of length four. 
	Furthermore, each large dotted box surrounds all configurations corresponding to the same orientation of an NCL edge $e=vw$. 
	Then, the set of configurations (non-dotted boxes) in each large dotted box induces a connected component; 
this means that any configuration in the set can be transformed into any other without changing the orientation of the corresponding NCL edge; 
this condition is called the ``internal connectedness'' of the gadget~\cite{DBLP:journals/ieicet/OsawaSIZ18}. 
	In addition, if we contract the configurations in the same large dotted box into a single vertex (and merge parallel edges into a single edge if necessary), then the resulting graph is exactly the graph depicted in \figurename~\ref{fig:edge_configuration}(a);
this condition is called the ``external adjacency'' of the gadget~\cite{DBLP:journals/ieicet/OsawaSIZ18}. 
	Therefore, we can conclude that our edge gadget correctly simulates the behavior of an NCL edge.

\medskip

\noindent
{\bf {\sc And} and {\sc or} gadgets.}
	\figurename~\ref{fig:gadget}(b) illustrates our {\sc and} gadget for each NCL {\sc and} vertex $v$,
where $e_1$ and $e_2$ are two weight-$1$ NCL edges and $e_a$ is the weight-$2$ NCL edge.
	\figurename~\ref{fig:and_configuration}(a) (in appendix) illustrates all valid orientations of the three edges incident to $v$, and \figurename~\ref{fig:and_configuration}(b) illustrates valid configurations of the {\sc and} gadget together with (images of) three edge gadgets for $e_1$, $e_2$, and $e_a$. 
	Then, as illustrated in \figurename~\ref{fig:and_configuration} (in the appendix), our {\sc and} gadget satisfies both ``internal connectedness'' and  ``external adjacency'', and  hence it correctly simulates an NCL {\sc and} vertex.


	\figurename~\ref{fig:gadget}(c) illustrates our {\sc or} gadget for each NCL {\sc or} vertex $v$, where $e_a$, $e_b$, and $e_c$ correspond to three NCL edges incident to $v$.
	For an NCL {\sc or} vertex, we need to forbid only one type of orientations of the three NCL edges: all NCL edges $e_a$, $e_b$, and $e_c$ are directed outward for $v$ at the same time, that is, all six connectors are covered by the edge gadgets for $e_a$, $e_b$ and $e_c$. 
	Our {\sc or} gadget forbids such a case, because otherwise we cannot cover the two (white) small vertices in the center.  
	In addition, this {\sc or} gadget satisfies both ``internal connectedness'' and  ``external adjacency'', and  correctly simulates an NCL {\sc or} vertex.

\paragraph*{Reduction.}

	As illustrated in \figurename~\ref{fig:subdivision}, we replace each of NCL edges and NCL {\sc and}/{\sc or} vertices with its corresponding gadget; 
let $G$ be the resulting graph.
	Notice that each of our three gadgets is of maximum degree three, and connectors in the edge gadget are of degree two; thus, $G$ is of maximum degree five. 
	In addition, each of our three gadgets is a bipartite graph such that two connectors in the same pair belong to different sides of the bipartition; therefore, $G$ is bipartite. 
	Furthermore, since NCL remains PSPACE-complete even if an input NCL machine is bounded bandwidth~\cite{DBLP:conf/iwpec/Zanden15}, the resulting graph $G$ is also bounded bandwidth and of maximum degree five;
notice that, since each gadget consists of only a constant number of edges, the bandwidth of $G$ is also bounded. 

	We next construct two perfect matchings of $G$ which correspond to two given NCL configurations $C_\ini$ and $C_\tar$ of the NCL machine. 
	Note that there are (in general, exponentially) many perfect matchings which correspond to the same NCL configuration. 
	However, by the construction of the three gadgets, no two distinct NCL configurations correspond to the same perfect matching of $G$. 
	We arbitrarily choose two perfect matchings $M_\ini$ and $M_\tar$ of $G$ which correspond
to $C_\ini$ and $C_\tar$, respectively. 

	This completes the construction of our corresponding instance of {\sc perfect matching reconfiguration}. 
	The construction can be done in polynomial time.
	Furthermore, the following lemma gives the correctness of our reduction. 
	\begin{lemma}[$\ast$] \label{lem:hardcorrect}
	There exists a desired sequence of NCL configurations between $C_\ini$ and $C_\tar$ if and only if there exists a reconfiguration sequence between $M_\ini$ and $M_\tar$.
	\end{lemma}

This completes the proof of Theorem~\ref{the:bipartite}.
\end{proof}

\subsection*{Remarks.}
	We conclude this section by giving some remarks that can be obtained from Theorem~\ref{the:bipartite}.
	We first prove that the problem remains intractable even for split graphs.
	A graph is {\em split} if its vertex set can be partitioned into a clique and an independent set. 
	\begin{corollary} \label{cor:split}
		\textsc{Perfect matching reconfiguration} is PSPACE-complete for split graphs.
	\end{corollary}
	\begin{proof}
		By Theorem~\ref{the:bipartite} the problem remains \PSPACE-complete for bipartite graphs. Consider the graph obtained by adding new edges so that one side of the bipartition forms a clique. The resulting graph is a split graph. 
		These new edges can never be part of any perfect matching of the graph. Indeed, since the original graph was bipartite, there must be the same number of vertices on each side of the bipartition. In a perfect matching of the split graph, all the vertices from the independent set must be matched with vertices from the clique, and no vertex from the clique remains to be matched together.
		Thus, the corollary follows. 
	\end{proof}

Let $k$ be an integer. An edge-subgraph $H$ of $G$ is a {\em $k$-factor} if all the vertices of $H$ have degree exactly $k$. 
Thus, a $1$-factor is a perfect matching. Theorem~\ref{the:bipartite} implies the following:

\begin{corollary}[$\ast$] \label{cor:kfactor}
	Let $G$ be a graph and $H_\ini,H_\tar$ be two $k$-factors. Deciding if there is a sequence of flip operations transforming $H_\ini$ into $H_\tar$ is PSPACE-complete. 
\end{corollary}

Finally, we show that it is \PSPACE-complete to decide whether two perfect
matchings are connected by a sequence of flip operations on alternating cycles
of given fixed length.

\begin{corollary}[$\ast$]
  For $k \geq 4$ and $k$ even, the problem $k$-\textsc{Perfect Matching Reconfiguration} is \PSPACE-complete.
  \label{cor:kflip}
\end{corollary}

%% file: interval.tex
\subsection{Strongly orderable graphs} \label{subsec:interval}

Interval graphs form easy instances for many NP-hard problems, and the situation is no different here. 
In fact, we prove that any instance on an interval graph is a \YES-instance. 
Our argument also yields a linear-time algorithm to compute a reconfiguration sequence of a linear number of flip operations between any two perfect matchings.

	For the sake of generality, we consider a wider class of graphs, called strongly orderable graphs. 
	A graph $G=(V,E)$ is \emph{strongly orderable} if there is a \emph{strong ordering} on its vertices, defined as follows:
an order $(v_1,v_2,\ldots,v_n)$ of $V$ such that for every $i,j,k,\ell$ with $i<j$ and $k<\ell$, if all of $v_iv_k, v_iv_\ell$ and $v_j v_k$ are edges, then $v_j v_\ell$ is an edge.
	Note that the class of strongly orderable graphs is hereditary: every induced subgraph of a strongly orderable graph is strongly orderable.

	Our proof strategy for the following theorem is to show that every perfect matching $N$ of a strongly orderable graph $G$ can be transformed into some particular perfect matching $M$ of $G$, called the canonical perfect matching;
then, any two perfect matchings $N$ and $N'$ of $G$ admit a reconfiguration sequence between them via $M$. 
	The \emph{canonical perfect matching} of a graph $G$ with respect to an order $\mathcal{O}=(v_1,v_2,\ldots,v_n)$ is a perfect matching of $G$ (if any) greedily obtained by selecting, among the available edges, the one with endpoints of smallest indices. 
%
	Note that any strongly orderable graph that admits a perfect matching, also admits a canonical perfect matching with respect to a corresponding order on the vertices, see e.g.~\cite{Dragan97}.
	We give the following theorem in this subsection. 
	\begin{theorem}[$\ast$] \label{the:stronglyorder}
		Let $G$ be a strongly orderable graph. 
		Then, there is a reconfiguration sequence of linear length between any two perfect matchings of $G$. 
		Furthermore, such a reconfiguration sequence can be found in linear time if we are given a strong ordering on the vertices of $G$ as a part of the input. 
	\end{theorem}

		The natural question regarding Theorem~\ref{the:stronglyorder} is whether a strong ordering can be computed efficiently. 
		In general, Dragan~\cite{dragan} proved that strongly orderable graphs $G=(V,E)$ can be recognized in $O(|V|\cdot(|V|+|E|))$ time, and if so we can obtain its strong ordering in the same running time. 
		However, when restricted to interval graphs, we can obtain a strong ordering in linear time~\cite{hsu}.	
		We thus have the following corollary. 
	\begin{corollary} \label{cor:interval}
	Let $G$ be an interval graph. 
	Then, there is a reconfiguration sequence of linear length between any two perfect matchings of $G$. 
	Furthermore, such a reconfiguration sequence can be found in linear time. 
	\end{corollary}

%

%% file: outerplanar.tex

\subsection{Outerplanar graphs} \label{subsec:outerplanar}

	In this subsection, we consider outerplanar graphs. 
	Note that there are $\NO$-instances for outerplanar graphs, e.g., induced cycles of length more than four. 
	Nonetheless, we give the following theorem.
	\begin{theorem} \label{the:outerplanar}
		\textsc{Perfect Matching Reconfiguration} can be decided in linear time for outerplanar graphs. 
		Furthermore, if a reconfiguration sequence exists, it can be found in linear time. 

		 Moreover, for a $\YES$-instance, a reconfiguration sequence of linear length can be output in linear time. 
	\end{theorem}

	We give such an algorithm as a proof of Theorem~\ref{the:outerplanar}.
	Suppose we are given a simple outerplanar graph $G=(V,E)$, and two perfect matchings $M_{\ini}$ and $M_{\tar}$ in $G$.
	We may assume that $G$ is connected as we can consider each connected component separately. 
	
	If $|V|=2$, then $M_{\ini} = M_{\tar} = E$, and hence the instance is trivially a $\YES$-instance. 
	Suppose that $G$ is not $2$-connected and has a cut vertex $v \in V$, that is, 
	$G-\{v\}$ consists of more than one connected component. 
	Since $|V|$ is even, there exists a vertex subset $X \subseteq V \setminus \{v\}$ inducing a connected component of $G - \{v\}$ 
	such that $|X|$ is odd. 
	Then, any perfect matching in $G$ contains an edge connecting $v$ and $X$. 
	This shows that we can consider two subgraphs $G_1:=G[X\cup \{v\}]$ and $G_2:=G - (X \cup \{v\})$, separately. 
	That is, we output ``$\YES$'' if $(G_i, M_{\ini} \cap E_i, M_{\tar}\cap E_i)$ is a $\YES$-instance for $i=1, 2$, where $E_i$ is the edge set of $G_i$,
	and output ``$\NO$'' otherwise. Thus, in what follows, we may assume that $G$ is $2$-connected. 

	Since $G$ is outerplanar and $2$-connected, all the vertices are on the outer boundary cycle. 
	Suppose that the vertices $v_1, v_2, \dots , v_n$ appear in this order along the cycle. 
	For simplicity, we denote $v_{n+1}=v_1$, $v_{n+2}=v_2$, and $v_0=v_n$. 
	If there exists a pair of indices $i, j \in \{1, 2, \dots , n\}$ such that $|i-j|$ is even and $v_i v_{j} \in E$, then we can remove $v_i v_j$ from $G$, 
	because it cannot be contained in any perfect matching of $G$. Indeed, the subgraph induced by $v_{i+1},\ldots,v_{j-1}$ is disconnected from the rest of the graph if $v_i,v_j$ are deleted and contains an odd number of vertices.
	In particular, after this change, $v_i v_{i+2} \not\in E$ for any $i \in \{1, 2, \dots , n\}$. 

	We now show the following lemma. 
	\begin{lemma}[$\ast$]\label{lem:degtwo}
	If $v_i v_{i+2} \not\in E$ for any $i \in \{1, 2, \dots , n\}$, then
	there exists an index $k \in \{1, 2, \dots , n\}$ such that both $v_k$ and $v_{k+1}$ have degree two.
	\end{lemma}

	Let $k$ be an index such that both $v_k$ and $v_{k+1}$ have degree two, and let $e= v_k v_{k+1}$. 
	We consider the following two cases separately. 
	\begin{description}
	\item[Case 1:] We first consider the case with $v_{k-1} v_{k+2} \not\in E$. 
		In this case, we can see that $e$ is not contained in any cycles of length four, and hence $e$ does not appear in the transformation.  
		Thus, if $v_k v_{k+1} \in M_{\ini} \triangle M_{\tar}$, then we can immediately conclude that 
		$(G, M_{\ini}, M_{\tar})$ is a $\NO$-instance. 
		If $v_k v_{k+1} \not\in M_{\ini} \cup M_{\tar}$, then we remove $e$ from the instance and repeat the procedure. 
		If $v_k v_{k+1} \in M_{\ini} \cap M_{\tar}$, then we remove $v_k$ and $v_{k+1}$ together with their incident edges from the instance and repeat the procedure. 
	\item[Case 2:] We next consider the case with $v_{k-1} v_{k+2} \in E$. 
		Note that if a perfect matching in $G$ does not contain $e$, then it has to contain both $v_{k-1} v_k$ and $v_{k+1} v_{k+2}$, 
		because $v_k$ and $v_{k+1}$ have degree two. 
		In this case, define $M'_{\ini} := M_{\ini} \setminus \{ e \}$ if $e \in M_{\ini}$ and
		$M'_{\ini} := (M_{\ini} \setminus \{v_{k-1} v_k, v_{k+1} v_{k+2} \}) \cup \{v_{k-1} v_{k+2} \}$ otherwise. 
		We also define $M'_{\tar} := M_{\tar} \setminus \{ e \}$ if $e \in M_{\tar}$ and
		$M'_{\tar} := (M_{\tar} \setminus \{v_{k-1} v_k, v_{k+1} v_{k+2} \}) \cup \{v_{k-1} v_{k+2} \}$ otherwise. 	
		Let $G' := G - \{v_k, v_{k+1}\}$. 
		Then, we solve a new smaller instance $(G', M'_{\ini}, M'_{\tar})$. 
	\end{description}
	In either case, we reduce the original instance to a smaller instance, 
   which shows that our algorithm runs in polynomial time. 
	The correctness of Case 1 is obvious. The correctness of Case 2 is guaranteed by the following lemma. 
	
	\begin{lemma}[$\ast$]\label{lem:reduction01}
		$(G, M_{\ini}, M_{\tar})$ is a $\YES$-instance if and only if $(G', M'_{\ini}, M'_{\tar})$ is a $\YES$-instance. 
	\end{lemma}

	By the above arguments, we obtain a polynomial-time algorithm for \textsc{Perfect Matching Reconfiguration} in outerplanar graphs.
	A pseudocode of our algorithm is shown in Algorithm~\ref{alg1}, which we denote $\AlgOP$ (that stands for Perfect Matching Reconfiguration in Outerplanar Graphs). 
	In the pseudocode, let \AlgOP($G, M_{\ini}, M_{\tar}$) denote the output of $\AlgOP$ when the input consists of $G, M_{\ini}$, and $M_{\tar}$. 
	Although this pseudocode simply outputs $\YES$ or $\NO$, it can be modified so that it actually finds a reconfiguration sequence. 

    To make the running time linear, we implement each step carefully and give the following lemma.
    \begin{lemma}[$\ast$] \label{lem:outerplanarlinear}
    	\textup{$\AlgOP$} can be implemented so that it runs in linear time. 
    \end{lemma}
	
	This completes the proof of Theorem~\ref{the:outerplanar}.\qed
	
	\begin{algorithm}[h]
	\SetKwInOut{Input}{Input}\SetKwInOut{Output}{Output}
	\Input{A simple outerplanar graph $G=(V,E)$, and two perfect matchings $M_{\ini}$ and $M_{\tar}$ in $G$.}
	\Output{``$\YES$'' if $M_{\ini} \sevstep{G} M_{\tar}$, and ``$\NO$'' otherwise.}
	\If{$|V|=2$}{
		Return ``$\YES$''.
	} 
	\If{$G$ is not $2$-connected} 
	{ \label{alg1:divide:graph:start}
		Divide $G$ into two graphs $G_1=(V_1, E_1)$ and $G_2=(V_2, E_2)$. \\  
		Return ``$\YES$'' if  \AlgOP($G_i, M_{\ini} \cap E_i, M_{\tar}\cap E_i$)=$\YES$ for $i=1, 2$, and return ``$\NO$'' otherwise. 
		\label{alg1:divide:graph:end}}
	Suppose that the vertices $v_1, v_2, \dots , v_n$ appear in this order on the boundary cycle. \\
	\While{There exists an edge $v_{i} v_{j} \in E$ such that $|i-j|$ is even}
	{ \label{alg1:remove:evenedge:start}
		Remove the edge $v_{i} v_{j}$ from $G$.
		\label{alg1:remove:evenedge:end} } 	
	Find an edge $e= v_k v_{k+1}$ such that both $v_k$ and $v_{k+1}$ have degree two.  \label{alg1:find:vk:vk1}\\
	\If{$v_{k-1} v_{k+2} \not\in E$}
	{
		\If{$e \in M_{\ini} \triangle M_{\tar}$}
		{Return ``$\NO$''.}
		\If{$e \not\in M_{\ini} \cup M_{\tar}$}
		{Return \AlgOP($G - e, M_{\ini}, M_{\tar}$).}
		\If{$e \in M_{\ini} \cap M_{\tar}$}
		{Return \AlgOP($G - \{v_k, v_{k+1}\}, M_{\ini} \setminus \{e\}, M_{\tar} \setminus \{e\}$).}
	}
	\If{$v_{k-1} v_{k+2} \in E$}
	{
		Define $M'_{\ini} := M_{\ini} \setminus \{ e \}$ if $e \in M_{\ini}$ and
		$M'_{\ini} := (M_{\ini} \setminus \{v_{k-1} v_k, v_{k+1} v_{k+2} \}) \cup \{v_{k-1} v_{k+2} \}$ otherwise. \\
		Define $M'_{\tar} := M_{\tar} \setminus \{ e \}$ if $e \in M_{\tar}$ and
		$M'_{\tar} := (M_{\tar} \setminus \{v_{k-1} v_k, v_{k+1} v_{k+2} \}) \cup \{v_{k-1} v_{k+2} \}$ otherwise. \\
		Return \AlgOP($G - \{v_k, v_{k+1}\}, M'_{\ini}, M'_{\tar}$).
	}
	\caption{\AlgOP}
	\label{alg1}
\end{algorithm}

%% file: cograph.tex
\subsection{Cographs} \label{subsec:cograph}
\renewcommand{\NP}{NP}
\renewcommand{\PSPACE}{PSPACE}
\newcommand{\Onotation}{O}
\newcommand{\symdiff}{\mathop{\bigtriangleup}}
\newcommand{\gmreconf}{\textsc{general matching reconfiguration}\xspace}

We now consider the complexity of \pmreconf when the input graph is a cograph. Cograph are graphs without a path on four vertices as an induced subgraph.

As examples concerning reconfiguration on this class of graphs, it is known
that the problems \textsc{independent set reconfiguration} and \textsc{Steiner
	tree reconfiguration} can be decided efficiently on
cographs~\cite{BB:14,Bonsma:16,MIZ:17}, while they are \PSPACE-complete for
general graphs~\cite{Ito:11,MIZ:17}. 
Theorem~\ref{the:bipartite} together with the following result show that the situation is similar for \pmreconf.

\begin{theorem} \label{the:cograph}
	\pmreconf\ on cographs can be decided in polynomial time. 
	Moreover, for a \YES-instance, a reconfiguration sequence of linear length can be output in polynomial time. 
\end{theorem}

We will use the following recursive characterization of cographs.
\begin{itemize}
	\item A graph consisting of a single vertex is a cograph.
	\item If $G$ and $H$ are cographs, then their disjoint union is a cograph, that is, the graph with the vertex set $V(G) \cup V(H)$ and the edge set $E(G) \cup E(H)$ is a cograph. 
	\item If $G$ and $H$ are cographs, then their complete join is a cograph, that is, the graph with the vertex set $V(G) \cup V(H)$ and the edge set $E(G) \cup E(H) \cup \{vw \mid v \in V(G), w \in V(H) \}$ is a cograph.
\end{itemize}
From this characterization of cographs, we can naturally represent a cograph $G$ by a binary tree, called a cotree of $G$, defined as follows:
a {\em cotree} $T$ of a cograph $G$ is a binary tree such that each leaf of $T$ is labeled with a single vertex in $G$, and each internal node of $T$ has exactly two children and is labeled with either ``union'' or ``join'' labels.
Such a cotree of a given cograph $G$ can be constructed in linear time~\cite{CPS:85}.
Each node of $T$ corresponds to a subgraph of $G$ which is induced by all vertices corresponding to all the leaves of $T$ that are the descendants of the node in $T$; thus, the root of $T$ corresponds to the whole graph $G$. 
The cotree of $G$ is not necessarily unique but the following two properties do not depend on the
choice of a cotree $T$. First, a non-trivial cograph is connected if and only if the root of 
$T$ is a join-node.  Furthermore, two vertices of a cograph
are joined by an edge if and only if their first common ancestor in $T$ is a join-node.

The main idea of the theorem is to decompose the graph, and apply the algorithm recursively on each of the components. In order to get a transformation of linear length using this method, we need to extend our problem to non-perfect matchings. Since the set of vertices matched by a matching does not change when performing a flip, we need to add some other operation. We will consider in this section that two matchings are adjacent if their symmetric difference is either a cycle of length four, or a path of length $3$. Note that this second type of transition we added corresponds to the token sliding model: we are allowed to replace an edge of the matching by any other incident edge. We will call this operation a \emph{sliding move}. We consider reconfiguration in this more general setting. \\

\vbox{
	\noindent\gmreconf \\
	\indent \textbf{Input:} Graph $G$, two matchings $M_{\ini}$ and $M_{\tar}$ of $G$.\\
	\indent \textbf{Question:} Is there a sequence of flips and sliding moves that transforms $M_{\ini}$ into $M_{\tar}$? \\
}

Note that the answer to the problem is clearly \NO when the two matchings do not have the same size. We can also remark that when the two input matchings are perfect, sliding moves become useless (since sliding requires at least one non-matched vertex), and we get back the original problem. In particular, Theorem~\ref{the:cograph} is a special case of the following more general result.
\begin{theorem} \label{the:gmcograph}
	\gmreconf on cographs can be decided in polynomial time. 
	Moreover, for a \YES-instance, a reconfiguration sequence of linear length can be output in polynomial time.
\end{theorem}

We start by considering certain base cases for which a transformation of linear length always exists and can be computed efficiently.
Let $M_\ini$ and $M_\tar$ be two matchings of a cograph $G$, and let $T$ be a cotree of
$G$. For the purpose of transforming $M_\ini$ to $M_\tar$, we may assume that $G$ is
connected, so the root of $T$ is a join-node (otherwise we can simply apply the algorithm to each connected component of $G$). We will call \emph{root partition} the partition of the vertices of $G$ into $A$ and $B$ corresponding to all the leaves in respectively the left and right subtrees of the root of $T$.
All along this section, unless otherwise specified, $A$ and $B$ will denote the root partition of $G$, and $k$ denotes the size of the matchings we are considering, i.e., $k = |M_\ini| = |M_\tar|$.
Note that $A$ is complete to $B$, because the root of $T$ is a join-node. 
Without loss of generality we may assume that $|A| \geq |B|$.
For a vertex subset $X$ of $G$ and an edge $e$ of $G$, we say that $G[X]$ \emph{contains} $e$ if both endpoints of $e$ 
are in $X$.

Given a connected cograph $G$ with root partition $A,B$, we will consider the two following conditions:

\begin{enumerate}[label=\textbf{(C\arabic*)}]
	\item\label{it+edge} There exists a matching $M$ of $G$ of size $k$ such that $G[B]$ contains an edge of $M$.
	
	\item\label{it+freevert} There exists a matching $M$ of $G$ of size $k$ such that at least one vertex of $B$ is not matched in $M$.
\end{enumerate}
When one of these two conditions holds, then we will be able to show directly that the reconfiguration graph on matchings of size $k$ is connected, and has linear diameter. On the other hand, if none of the two condition holds, we will apply induction on $G[A]$ in order to conclude. The case where condition \ref{it+edge} holds is treated in Lemma~\ref{lem+cclifedge} and condition~\ref{it+freevert}  is handled in Lemma~\ref{lem+freeB}. The case where none of the two properties hold is treated in Lemma~\ref{lem+nofreenoedge}.

Note that it is possible to check in polynomial time whether one of the two conditions holds since computing a maximum matching in a graph can be done in polynomial time. Indeed, a simple algorithm to check condition~\ref{it+edge} consists in trying all possible edges in $G[B]$, removing both endpoints from the graph, and search for a matching of size $k-1$ in the remaining graph. Similarly, for condition~\ref{it+freevert}, we can try to remove every possible vertex from $B$, and search for a matching of size $k$ in the remaining graph.

Remark that the connected components of $M_\ini \symdiff M_\tar$ can be of three types: single edges, paths of length at least $3$, and even cycles. If $M_\ini$ and $M_\tar$ are perfect matchings, then only even cycles can occur in the symmetric difference. We start with the following observation that finding a reconfiguration sequence is easy if the symmetric difference contains no cycle.

\begin{lemma}
	\label{lem+cyclefree}
	Let $G$ be a cograph, and $M_\ini$ and $M_\tar$ be two matchings of size $k$ such that $M_\ini \symdiff M_\tar$ contains no cycle. Then there is a transformation of length at most $2|M_\ini \symdiff M_\tar|$ from $M_\ini$ to $M_\tar$.
\end{lemma}
\begin{proof}
	The result is proved by induction on the size of the symmetric difference $M_\ini \symdiff M_\tar$. If the symmetric difference is zero, then $M_\ini = M_\tar$ and the result trivially holds. Otherwise, we only need to show that we can reduce the symmetric difference by $2$ in at most $4$ steps. 
	First assume that $M_\ini \symdiff M_\tar$ contains a path of length $t\geq 3$. Let $x_1, \ldots x_t$ be the vertices of this path, and assume without loss of generality that $x_1x_2$ is an edge in $M_\ini$ and $x_2 x_3$ an edge in $M_\tar$. By definition, $x_1$ is not matched in $M_\tar$. Consequently, starting from $M_\tar$ we can slide the edge $x_2x_3$ to $x_1x_2$, and reduce the symmetric difference by $2$ in one step.

	Now assume that the symmetric difference does not contain any path of length at least $3$. Since by assumption it does not contain any cycles either, then it must be a disjoint union of edges. Let $e_\ini=u_\ini v_\ini$ be an edge in $M_\ini\setminus M_\tar$, and $e_\tar=u_\tar v_\tar$ an edge in $M_\tar \setminus M_\ini$ (none of these sets is empty since $M_\tar$ and $M_\ini$ have the same size).  First assume that there is an edge incident to both~$e_\ini$ and~$e_\tar$. Without loss of generality, we can assume that this edge is $u_\ini u_\tar$. In this case, we can simply slide $u_\ini v_\ini$ to $u_\ini u_\tar$ and then to $u_\tar v_\tar$, giving a transformation of length at most $2$.
	Otherwise, since $G$ is a cograph, the two edges must be at distance at most $2$. Hence, we can assume that there is a vertex $w$ adjacent to both $u_\ini$ and $u_\tar$. We consider the two following cases:
	\begin{itemize}
		\item $w$ is not matched in $M_\ini$. In this case, starting from $M_\ini$, we can simply slide $u_\ini v_\ini$ to $u_\ini w$ and then to $wu_\tar$ and finally to $u_\tar v_\tar$;
		\item $w$ is matched to $w'$ in $M_\ini$. In this case, we start by sliding $ww'$ to $wu_\tar$ and then to $u_\tar v_\tar$. Then, we can slide $u_\ini v_\ini $ to $u_\ini w$ and finally $ww'$.
	\end{itemize}
	In any case, we can reduce the symmetric difference by $2$ in at most $4$ steps, which proves the result.
\end{proof}

The following lemma provides a way to construct a transformation sequence of linear length when condition~\ref{it+edge} holds. Note that the proof is constructive and can be easily turned into a polynomial time algorithm computing this sequence.

\begin{lemma}
	\label{lem+cclifedge}
	Let $G$ be a connected cograph with $n$ vertices, and $k \geq 0$ such that condition~\ref{it+edge} holds. 
	Then, there is a reconfiguration sequence of length $O(n)$ between any two matchings of size $k$ in $G$.
\end{lemma}

\begin{proof}
	Assume that $G$ has a perfect matching $M$ such that $G[B]$ contains at least on edge $M$, and let $e$ be such an edge. To prove the lemma, we only need to show that 
	for any matching $M_\ini$ of size~$k$, there is a transformation sequence of linear length from $M_\ini$ to $M$.

	We first claim the following: we can assume without loss of generality that $e$ is the only edge of $M$ contained in $B$. 
	To see this, assume that there is another edge $e' = ab$ with $a,b \in B$. Since $|A| \geq |B|$, either $M$ also contain an edge $e'' = cd$ with both endpoints in $A$, or there is at least one vertex $x$ in $A$ not matched in $M$. In the first case, the edge $e'$ can be removed from the matching by flipping $ab$ and $cd$ with $ac$ and $bd$. In the second case, we can slide $e' = ab$ into $ax$.
	
	We now prove that for every matching $M_\ini$ of $G$ of size $k$, there is a reconfiguration sequence of length $O(|M \symdiff M_\ini|)$ from $M$ to $M_\ini$. 
	We start by proving the following claim. See Figure~\ref{fig+claimIllustration} for an illustration of some of the cases.
	
	\begin{figure}
		\includegraphics[width=\textwidth]{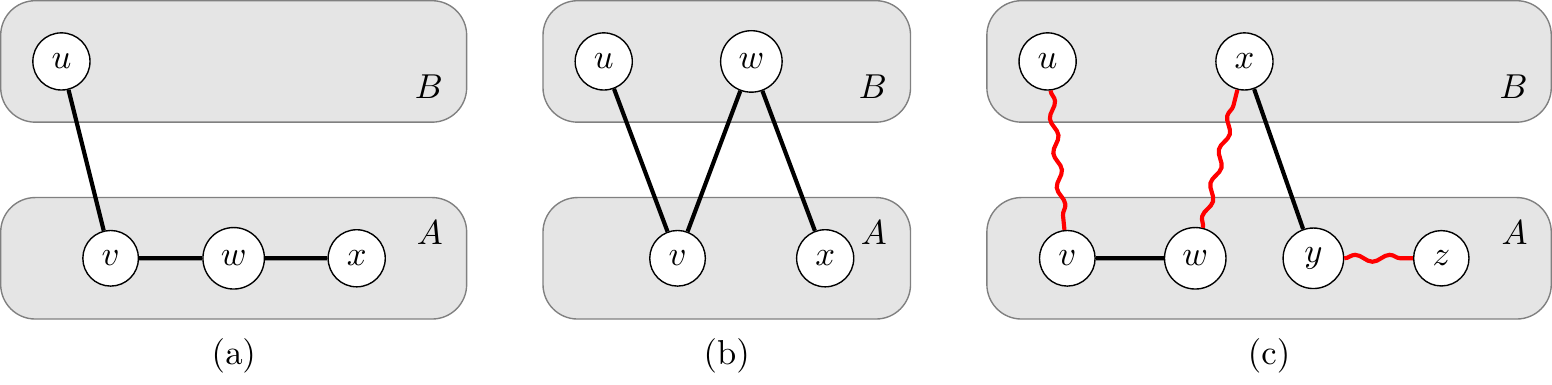}

		\caption{\label{fig+claimIllustration} Representation of some of the cases of Claim~\ref{clm+cograph+enum}. Note that all edges between $A$ and $B$ are present in the graph, but were removed for clarity. The edges drawn are the edges in the symmetric difference $M \symdiff M_\ini$. (a) corresponds to Assumption~\ref{it+bapath}, (b) is Assumption~\ref{it+zigzagpath}, and (c) is Assumption~\ref{it+wpath}. For the last case, the red wavy edges are edges in $M$, the other ones are edges in~$M_\ini$.} 
	\end{figure}

	\begin{claim}
		\label{clm+cograph+enum}
		After a transformation of at most $O(|M \symdiff M_\ini|)$ steps in $M$ and $M_\ini$, we can assume\footnote{The resulting matchings are still denoted by $M$ and $M_\ini$ for simplicity.} that $M$ still contains an edge in $G[B]$ and $M$ and $M_\ini$ also satisfy the following statements:
		\begin{enumerate}
			\item no edge of $M_\ini$ is contained in $B$. \label{it+sameedge}
			\item $M \symdiff M_\ini$ contains no cycle on $A$. \label{it+noacycle}
			\item $M \symdiff M_\ini$ contains no three edges $uv$, $vw$, $wx$, such that $u \in B$ and $v, w, x \in A$. \label{it+bapath}
			\item $M \symdiff M_\ini$ contains no three edges $uv$, $vw$, $wx$, each between $A$ and $B$. \label{it+zigzagpath}
			\item $M \symdiff M_\ini$ contains no five edges $uv$, $vw$, $wx$, $xy$ and $yz$, with $uv \in M$, and $v, w, y, z \in A$ and  $u, x \in B$. \label{it+wpath}
		\end{enumerate}
	\end{claim}
	
	\begin{proof}[Proof of Claim~\ref{clm+cograph+enum}]
		We prove all the points in the increasing order. Let $e$ be the edge of $M$ contained in $G[B]$, and let $x_e$ and $y_e$ be its two endpoints.
		
		\paragraph{Assumption~\ref{it+sameedge}.} Let $ab$ be an edge of $M_\ini$ contained in $B$. Since $|A| \geq |B|$, there is either an edge $cd$ of $M_\ini$ with both endpoints in $A$ or a vertex $x$ in $A$ not matched by $M_\ini$. In the first case, starting from $M_\ini$, we flip $ab$ and $cd$ with $bc$ and $ad$. In the second case, we make a sliding move from $ab$ to $ax$. Since every edge in $M_\ini$ contained in $B$ is in the symmetric difference $M_\ini \symdiff M$ (except if $ab = e$), this operation will be repeated at most $|M_\ini \symdiff M| + 1$ times. Each time, the symmetric difference increases by at most $1$ (by $2$ in the case $e = ab$).

		\paragraph{Assumption~\ref{it+noacycle}.} Let $C$ be a cycle in $M \symdiff M_\ini$,
		such that $V(C) \subseteq A$. We show that we can transform $M \cap E(C)$ to
		$M_\ini \cap E(C)$. The other edges of $M$ are not modified.
		Let $u_1,u_2,\ldots,u_{2\ell}$ be the vertices of $C$, all in $A$. We assume without loss
		of generality that $u_2u_3,u_4u_5,\ldots,u_{2\ell}u_1 \in M$.  We first
		flip $x_ey_e$ and $u_{2\ell}u_1$ to create $x_eu_1$ and $y_eu_{2\ell}$. 
		Then we flip $x_eu_1$ and $u_2u_3$ for $u_1u_2$ and $x_eu_3$.  
		We proceed with $u_4u_5$ and
		$uu_3$ and so on (reducing the length of the cycle in the symmetric difference), 
		until $x_eu_{2\ell-1}$ and $y_eu_{2\ell}$ remains. After
		flipping these two edges for $x_e y_e$ and $u_{2\ell-1}u_{2\ell}$, the resulting matching
		still contains the edge $x_ey_e$ with both endpoints in $B$ and we have reconfigured 
		$M \cap E(C)$ to $M_\ini \cap E(C)$, i.e., the size of the symmetric difference has decreased. The other edges in $M$ were not modified; in particular Assumption~\ref{it+sameedge} still holds.

		Note that number of flips performed in
		this sequence is at most $|E(C)|$ and the symmetric difference decreased by $|E(C)|$. \smallskip
		
		\paragraph{Assumption~\ref{it+bapath}.} We can
		assume without loss of generality that $uv$ and $wx$ are in $M$ and $vw$ is in $M_\ini$.
		Then, we can flip $uv$ and $wx$ for $vw$ and $ux$ and reduce the size of the symmetric difference by $2$.
		Moreover, Assumption~\ref{it+sameedge} 
		still holds in the resulting perfect matching.\smallskip
		
		\paragraph{Assumption~\ref{it+zigzagpath}.} We can
		assume without loss of generality that $uv$ and $wx$ are in $M$ and $vw$ is in $M_\ini$. Then, in $M$
		we can flip $uv$ and $wx$ for $vw$ and $ux$ and reduce the size of the symmetric difference.
		Note moreover that Assumption~\ref{it+sameedge} 
		still holds in the resulting perfect matching.\smallskip
		
		\begin{figure}
			\centering
			\includegraphics[width=\textwidth]{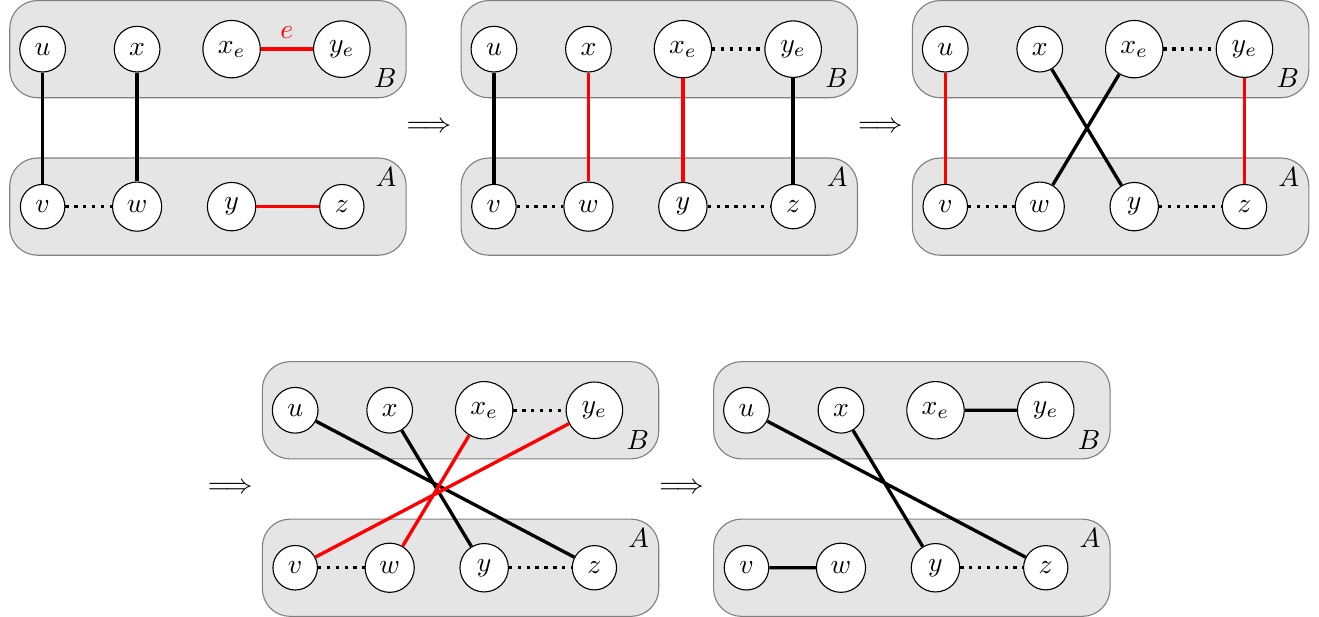}

			\caption{\label{fig:cograph:wpath} Reconfiguration sequence to reduce the symmetric difference in the case of Assumption~\ref{it+wpath}. All the edges between $A$ and $B$ are present in the graph but were removed for clarity. The edges drawn are the edges in $M$. At each step, the two edges in red are the ones which are flipped to obtain the next transition.}
		\end{figure} 
		
		\paragraph{Assumption~\ref{it+wpath}.}
		By assumption, $uv$, $wx$ and $yz$ are edges of $M$.
		Note that $x_e$ and $y_e$ are distinct from $\{ u,v,w,x,y,z \}$ since $u$ and $x$ are incident to edges of $M$ between $A$ and $B$ (and the other vertices are in $A$). We will perform a sequence of flips decreasing the symmetric difference, and preserving $e$ at the end of the transformation (see Figure~\ref{fig:cograph:wpath}).
		Now, in the matching $M$, flip $x_ey_e$ and $yz$ for $x_ey$ and $y_ez$. Then flip
		$xw$ and $x_ey$ for $xy$ and $x_ew$. Note that at this point, the size of the symmetric
		difference with $M_\ini$ decreased by one. Then we flip $y_ez$  and $uv$ for $y_ev$
		and $uz$. Finally, we flip $x_ew$ and $y_ev$ for $x_ey_e$ and $vw$. By these
		operations, the size of the symmetric difference with $M_\ini$ has decreased by
		at least two since we only flip edges of the symmetric difference and the resulting
		matching have edges $vw$ and $xy$ which are in $M_\ini$. Moreover, the resulting matching 
		still contains the edge $e$ with both endpoint in $B$.
		
		\paragraph{Number of flips.} In order to get Assumption~\ref{it+sameedge}, we may need $|M \symdiff M_\ini| +1$
		steps and increase the symmetric difference by $|M \symdiff M_\ini| +2$. 
		In all the other points,
		if we perform $\delta$ flips, we decrease the symmetric difference by at least $c \delta$ (where $c$
		is some constant), and we never have to apply Assumption~\ref{it+sameedge}
		again. So the claim holds after $O(|M \symdiff M_\ini|)$ steps and the size of the symmetric difference
		is still at most $O(|M \symdiff M_\ini|)$.
	\end{proof}
	
	After applying Claim~\ref{clm+cograph+enum}, let us still denote by $M$ and $M_\ini$ the resulting matchings. The number of steps needed to reach this point is at most $O(|M \symdiff M_\ini|)$. Recall that $e$ is the edge of $M$ contained in $B$. We will show that the only cycle that can be in the symmetric difference is a cycle of length $4$ containing the edge $e$. Assume by contradiction that this is not the case, and let $C$ be a cycle in the symmetric difference $M \symdiff M_\ini$. 
	
	Suppose first that $C$ does not contain $e$. By Assumption~\ref{it+noacycle}, there exists a vertex $u \in C \cap B$. Let $v$, $w$ and $x$ the vertices following $u$ on the cycle $C$ and such that $uv \in M$. We must also have $wx \in M$ and $vw \in M_\ini$. Since $C$ does not contain $e$, we know that $v \in A$. By Assumptions~\ref{it+bapath} and~\ref{it+zigzagpath}, we have $w \in A$, and $x \in B$. Since both $u$ and $x$ are on the side $B$, and using Assumption~\ref{it+sameedge}, we have $xu \not \in C$. In particular $|C| > 4$, and since $C$ has even length this implies $|C| \ge 6$. Let $y$ and $z$ be the two vertices following $x$ on the cycle $C$. By Assumption~\ref{it+sameedge}, and since $wx \in M$, we have $y \in A$. Moreover, by Assumption~\ref{it+zigzagpath}, we have $z \in A$. However, in this case the configuration of the vertices $u,v,w,x,y,z$ contradicts Assumption~\ref{it+wpath}.
	
	Consequently, there is only one cycle $C$ in the symmetric difference, and $C$ must contain $e$. Assume by contradiction that $|C| \geq 6$. Let $x_1, x_2, x_3, u, v, w$ be consecutive vertices on the cycle $C$ such that $x_1$ is incident to $e$, and $x_2$ is not. Then $x_1x_2 \in M_\ini$, which implies that $x_2 \in A$. Using the Assumptions~\ref{it+sameedge},~\ref{it+bapath}, and~\ref{it+zigzagpath}, we also have $x_3, v, w \in A$, and $u \in B$. In particular,~$w$ is not incident to~$e$. Let $x_0$ be the other endpoint of $e$. Since the cycle $C$ must have even length, $w$ and $x_0$ are not consecutive in $C$. Let $x$ be the vertex consecutive to $w$ in $C$. Then, by Assumption~\ref{it+bapath} we must have $x \in B$. Since $M_\ini$ does not contain any edge in $B$, this implies that $x$ and $x_0$ are not consecutive in $C$. Let $y$ be the vertex consecutive to $x$ in $C$, and $z$ consecutive to $y$. Since $C$ has even length, we known that $z \neq x_0$. By Assumption~\ref{it+sameedge} we must have $y \in A$, and by Assumption~\ref{it+zigzagpath}, we have $z \in A$. However, in this case the configuration of the vertices $u,v,w,x,y,z$ contradicts Assumption~\ref{it+wpath}.
	
	Hence, the only possible cycle in the symmetric difference is a cycle of length $4$ containing~$e$. After flipping this cycle, the symmetric difference contains only paths and isolated edges. By Lemma~\ref{lem+cyclefree}, we can finish transforming $M_\ini$ into $M$ using an additional $O(M \Delta M_\ini)$ steps.
\end{proof}

We now handle the case where condition~\ref{it+freevert} holds. As for the previous case, when this condition holds a transformation sequence can be easily found between any two matchings.

\begin{lemma}
	\label{lem+freeB}
	Let $G$ be a cograph, and $k \geq 0$ such that condition~\ref{it+freevert} holds. Then, there is a transformation sequence of length $O(n)$ between any two matchings of size $k$.
\end{lemma}
\begin{proof}
	Without loss of generality, we can assume that condition~\ref{it+edge} does not hold since otherwise we can conclude directly using~Lemma~\ref{lem+cclifedge}. Consequently, no matching of size $k$ of $G$ uses any of the edges in $G[B]$. Hence, these edges can be removed without changing in any way the reconfiguration graph, and we can assume that $G[B]$ is an independent set.
	
	Let $M$ be a matching of $G$ of size $k$ such that there exists a vertex $v$ in $B$ not matched in $M$. Let $M_\ini$ be any matching of $G$ of size $k$. We will show that there is a transformation from $M$ to $M_\ini$ of length at most $O(|M \symdiff M_\ini|)$. If the symmetric difference is zero, then the result is trivial. If the symmetric difference contains no cycle, then the result follows from Lemma~\ref{lem+cyclefree}. 
	
	Let $C$ be a cycle in the symmetric difference $M\symdiff M_\ini$. Since $G[B]$ is an independent set, $C$ must contain at least one vertex in $A$. Let $x_1, \ldots, x_{2t}$ be the vertices in $C$, with $x_1 \in A$, and $x_1x_2 \in M$. We consider the following transformation starting from the matching $M$:
	\begin{itemize}
		\item slide $x_1x_2$ to $x_1v$,
		\item for $i$ from $2$ to $t-1$ slide $x_{2i+1}x_{2i+2}$ to $x_{2i}x_{2i+1}$,
		\item finally, slide $x_1v$ to $x_{2t}x_1$.
	\end{itemize}
	This operation progressively reconfigure $M$ such that $M$ and $M_\ini$ agree on $C$. Note that the edges of $M$ outside of $C$ are not modified by the transformation. Moreover, if $M_2$ is the matching obtained after the transformation, then the symmetric difference with $M_\ini$ has decreased by $|C| = 2t$. Additionally, the vertex $v$ is still not matched in $M_2$. Since the number of steps performed by this transformation is $t+1 \leq |C|$, the result follows by applying induction with $M_2$ and $M_\ini$.
\end{proof}


In case none of the two conditions~\ref{it+edge} and~\ref{it+freevert} holds, the following lemma states that we only need to consider what happens on the subgraph induced by $A$. Given a matching $M$, we will note $M^A$ the matching of $G[A]$ induced by the edges of $M$. Remark that even if $M$ is a perfect matching of $G$, $M^A$ might not be a perfect matching of $G[A]$ since some of the vertices in $A$ can be matched to vertices in $B$ by $M$. This is the main reason we had to extend the problem to non-perfect matchings. We have the following.

\begin{lemma}
	\label{lem+nofreenoedge}
	Let $G$ be a connected cograph with root partition $A,B$ with $|A| \geq |B|$, and $k \geq 0$ such that conditions~\ref{it+edge} and~\ref{it+freevert} do not hold. Let $M_\ini$ and $M_\tar$ be two matchings of $G$ of size $k$. Then:
	\begin{itemize}
		\item there is a transformation sequence from $M_\ini$ to $M_\tar$ in $G$ if and only if there is a transformation sequence from $M_\ini^A$ to $M_\tar^A$ in $G[A]$;
		\item if there is a transformation of length $t$ from $M_\ini^A$ to $M_\tar^A$ in $G[A]$, then there is a transformation from $M_\ini$ to $M_\tar$ of length at most $t + O(|B|)$.
	\end{itemize}
\end{lemma}
\begin{proof}
	Let $M_\ini$ and $M_\tar$ be two matchings of size $k$ of $G$. 
	First assume that there is a transformation sequence $S$ from $M_\ini^A$ to $M_\tar^A$ in $G[A]$. We will build a transformation sequence from $M_\ini$ to $M_\tar$ in $G$. First, observe that any flip in $S$ on the subgraph $G[A]$ is also a valid flip on the whole graph $G$. For sliding moves, there are two possibilities. Let $u, v, w$ be three vertices in $A$, and consider the sliding move in $G[A]$ which replaces $uv$ by $vw$. Either $w$ is not matched to a vertex in $B$, and in this case this move is also a valid sliding move on the whole graph. Or $w$ is matched to a vertex $x \in B$. In this case, consider the operation of flipping $uv$ and $wx$ for $vw$ and $xu$. Then this transformation acts exactly as the original sliding move on $A$.

	Hence, by eventually replacing some of the sliding moves by flips as explained above, we obtain a transformation sequence $S'$ which transforms $M_\ini$ into a matching $M$ such that $M^A$ and $M_\tar^A$ are equal. To finalize the transformation, from $M$ to $M_\tar$, we only need the two following observation.
	\begin{itemize}
		\item We can make $M$ and $M_\tar$ agree on the vertices $a \in A$ which are matched to vertices in $B$. Indeed, if there is a vertex $a \in A$ which is matched to $b \in B$ in $M$ but not in $M_\tar$, then there must be a vertex $a' \in A$ which is matched in $M_\tar$ but not in $M$ (since $|M| = |M_\tar|$). Then, we can simply slide~$ab$ to~$a'b$.
		
		\item Let $A'$ the set of vertices in $A$ which are matched to vertices in $B$ in $M_\ini$ (and by the point above, also in $M_\tar$), and $G'$ the complete bipartite graph between $A$ and $B$. Note that $G'$ is a subgraph of $G$. The edges in $M_\ini$ (and $M_\tar$) in $G'$ form a perfect matching of $G'$. These perfect matchings can be seen as a permutation on the vertices of $B$. A flip in this graph consists in applying a transposition to the permutation. Hence, transforming $M$ into $M_\tar$ is equivalent to transforming one permutation into an other using transposition. It is well known that this is always possible using at most $|B|$ transpositions.
	\end{itemize}
	
	Hence, if there is a transformation of length $t$ from $M^A_\ini$ to $M^A_\tar$, then there is a transformation of length $t + O(|B|)$ from $M_\ini$ to $M_\tar$.
	
	Conversely, assume that there is a transformation sequence from $M_\ini$ to $M_\tar$. We want to show that there is a transformation sequence from $M_\ini^A$ to $M_\tar^A$. For this, we only need to show that if there is a one step transformation between $M_\ini$ and $M_\tar$, there is a one step transformation between $M_\ini^A$ and $M_\tar^A$. We consider the symmetric difference $D = M_\ini \Delta M_\tar$. If $D$ contains no edge in $G[A]$, then $M_\ini^A$ and $M_\tar^A$ are equal, and there is nothing to prove. Similarly, if $D \subset G[A]$, then $M_\ini^A$ and $M_\tar^A$ are adjacent by definition. Thus, we can assume in the following that none of these two cases happen.
	
	If $D$ is a path of length $3$ (i.e., the transformation is a sliding move). By the remarks above, we can assume that $D$ contains one edge in $G[A]$, but not the other. Let $u,v,w$ be the vertices of $D$, with $u,v \in A$ and $w \in B$. Then $w$ is not matched in one of $M_\ini$ or $M_\tar$. This contradicts the assumption that $G$ does not satisfy condition~\ref{it+freevert}. 
	
	If $D$ is a cycle of length $4$. There are two possible sub-cases:
	\begin{itemize}
		\item $D$ contains exactly on edge in $G[A]$. In this case $D$ must also contain one edge in $G[B]$, and this implies that one of $G_\ini$ or $G_\tar$ contains an edge in $G[B]$, a contradiction of the assumption that $G$ does not satisfy the condition~\ref{it+edge}.
		
		\item $D$ contains exactly two edges in $G[A]$. These two edges must be incident since otherwise $D \subset G[A]$. Then, this means that $M_\ini^A \Delta M_\tar^A$ is a path of length $3$ and the two matchings are adjacent in the reconfiguration graph.
	\end{itemize}
	If $D$ has three edges in $G[A]$, then we must have $D \subseteq G[A]$, and this case was already handled above. 
	
	Hence, in any case, if there is a transformation from $M_\ini$ to $M_\tar$, there is also a transformation from $M_\ini^A$ to $M_\tar^A$. This shows the reverse implication and ends the proof of the lemma.
\end{proof}

\begin{proof}[Proof of Theorem~\ref{the:gmcograph}]
	Given a cograph $G$ and two matchings $M_\ini$ and $M_\tar$ of $G$, the algorithm proceeds as follows: 
	\begin{enumerate}
		\item \label{step1} If $|M_\ini| \neq |M_\tar|$, then return \NO, otherwise let $k = |M_\ini| = |M_\tar|$
		\item \label{step2} If $G$ is not connected, call recursively the algorithm on each connected component. Otherwise let $A,B$ be the root
		partition of $G$, with $|A| \geq |B|$.
		\item \label{step3} If $G$ satisfies one of the conditions~\ref{it+edge} and~\ref{it+freevert}, output \YES, and produce a transformation sequence using either Lemma~\ref{lem+cclifedge} or Lemma~\ref{lem+freeB}.
		\item \label{step4} If $G$ does not satisfy any of these conditions, call recursively the algorithm on $A$, and decide the instance (and produce a transformation if it exists) using Lemma~\ref{lem+nofreenoedge}.
	\end{enumerate}
	
	Let us show that this algorithm is correct, runs in polynomial time and produces a transformation sequence of linear length.
	
	\paragraph{Correctness.} Whenever the algorithm answers \YES, it also provides a certificate (i.e., a transformation sequence). Hence, the only case where the algorithm might be incorrect is when it answers \NO. Let us show by induction on $G$ that if the algorithm returns \NO\ on a cograph $G$ given two matchings $M_\ini$ and $M_\tar$ as input, then no transformation exists between the two matchings.
	If the algorithm returns \NO\ in step~\ref{step1}, then there is clearly no transformation. If one of the recursive calls returns \NO\ in step~\ref{step2}, then there is trivially no transformation either. Finally, if one of the recursive calls returns \NO\ in step~\ref{step4}, then there is also no transformation sequence by Lemma~\ref{lem+nofreenoedge}.

	\paragraph{Running-time.} At each recursive call, the algorithm only performs a polynomial number of steps. Indeed, checking whether $G$ is connected, and comparing the size of $M_\ini$ and $M_\tar$ can be done in polynomial time. Additionally, computing cotree of a cograph can be done in polynomial time, and as we mentioned before, the two conditions~\ref{it+edge} and~\ref{it+freevert} can be verified in polynomial time. Finally, all the recursive calls are made on vertex disjoint subgraphs. Consequently, it follows immediately that the algorithm runs in polynomial time. 
	
	\paragraph{Length of transformation.} Let $G$ be a cograph, and $M_\ini$ and $M_\tar$ two matchings of $G$ such that there is a transformation from $M_\ini$ to $M_\tar$. Let us show by induction on $G$ that the algorithm produces a transformation of length at most $Cn$ for some constant $C$. 
	
	If the algorithm returns at step~\ref{step2}, then by induction it produces a transformation of length at most $C n_i$ on each component $G_i$ of $G$, with $n_i = |G_i|$. By combining the transformations on each component, we obtain a transformation for the whole graph of length at most $C(\sum n_i) = C n$.
	
	If the algorithm returns \YES\ during step~\ref{step3}, then the induction step directly follows from Lemmas~\ref{lem+cclifedge} and~\ref{lem+freeB}, provided $C$ is chosen large enough.
	
	Finally, if the algorithm outputs \YES\ at step~\ref{step4}, then using the induction hypothesis on $A$, it produces a transformation of length at most $C|A|$ from $M_\ini \cap G[A]$ to $M_\tar \cap G[A]$. By Lemma~\ref{lem+nofreenoedge}, the total sequence produced by the algorithm has length at most $C|A| + C'|B| \leq Cn$ where $C'$ is the constant in the big-Oh notation of Lemma~\ref{lem+nofreenoedge} and assuming $C \geq C'$.
\end{proof}

%% file: conclusion.tex
\section{Conclusion}

We introduced the \pmreconf problem and analyzed its complexity from the
viewpoint of graph classes. We showed that this problem is
\PSPACE-complete on split graphs and bipartite graphs of bounded bandwidth and
maximum degree five. Furthermore, we gave polynomial-time algorithms for
strongly orderable graphs, outerplanar graphs, and cographs. Each of the
algorithm outputs a reconfiguration sequence of linear length in polynomial
time. 

A natural open question is on which graph classes a \emph{shortest}
reconfiguration sequence can be found in polynomial time.  Furthermore, it would be
interesting to investigate if the flip operation can be used in order to sample
perfect matchings uniformly.

%% file: appendix.tex
\section{A figure and Proofs omitted from Section~\ref{sec:hard}}

\begin{figure}[h]
	\begin{center}
		\includegraphics[width=\linewidth]{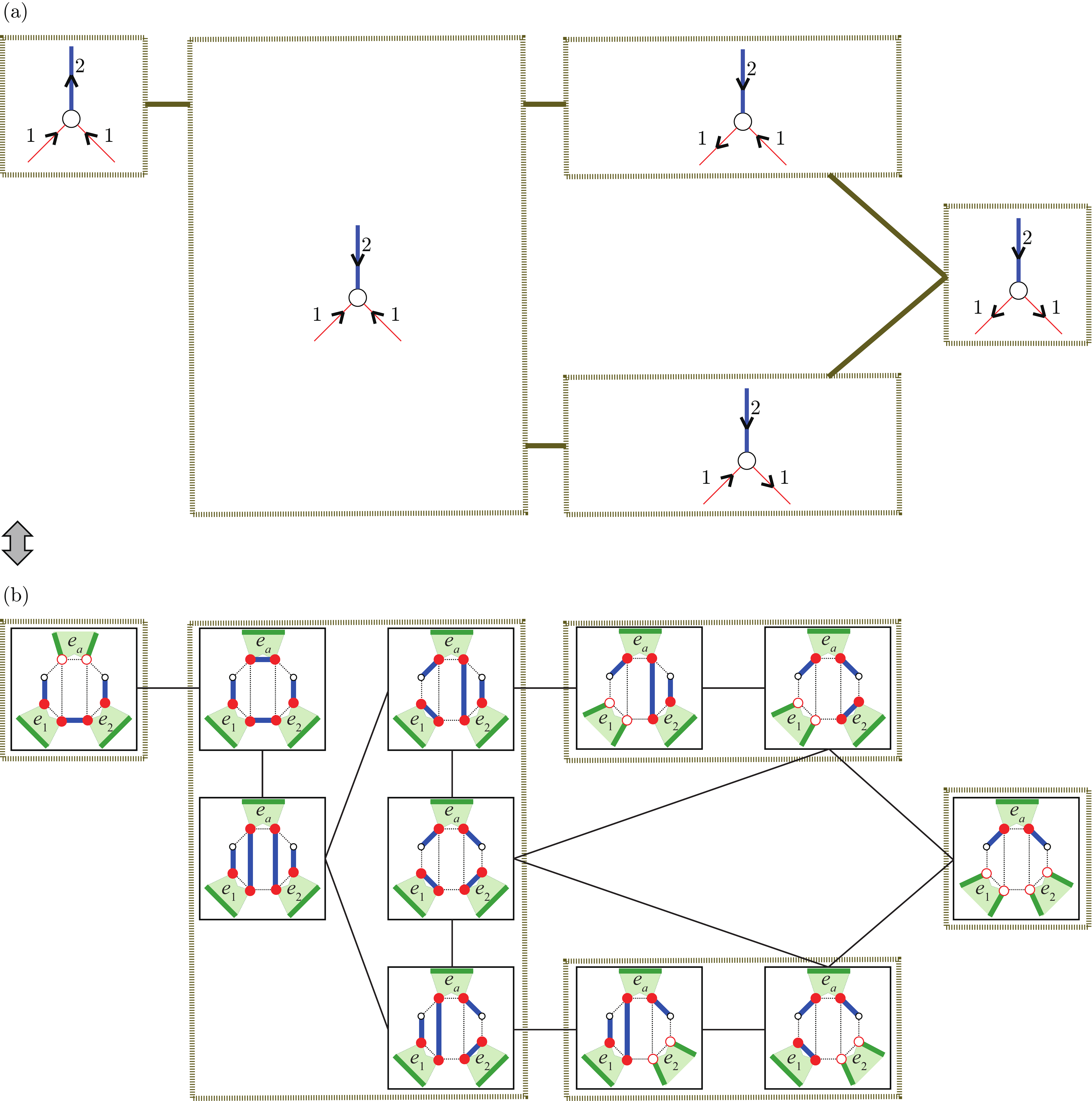}
	\end{center}
	\caption{(a) All valid orientations of three edges incident to an NCL {\sc and} vertex $v$, and (b)~all configurations of the {\sc and} gadget together with three incident edge gadgets, where the inside of each connector is painted (by red) if it is matched by the {\sc and} gadget.}
	\label{fig:and_configuration}
\end{figure}

\subsection{Proof of Lemma~\ref{lem:hardcorrect}}
\begin{proof} 
	We first prove the only-if direction.
	Suppose that there exists a desired sequence of
	NCL configurations between $C_\ini$ and $C_\tar$, and consider any two adjacent NCL configurations $C_{i-1}$ and $C_i$ in the sequence.
	Then, only one NCL edge $vw$ changes its orientation between $C_{i-1}$ and $C_i$.
	Notice that, since both $C_{i-1}$ and $C_i$ are valid NCL configurations, the NCL {\sc and}/{\sc or} vertices $v$
	and $w$ have enough in-coming arcs even without $vw$.
	Therefore, we can simulate this reversal
	by the reconfiguration sequence of perfect matchings in \figurename~\ref{fig:edge_configuration}(b) which passes through the
	neutral orientation of $vw$ as illustrated in \figurename~\ref{fig:edge_configuration}(a).
	Recall that both {\sc and} and {\sc or} gadgets
	are internally connected, and preserve the external adjacency. Therefore, any reversal of an
	NCL edge can be simulated by a reconfiguration sequence of perfect matchings of $G$, and hence
	there exists a reconfiguration sequence between $M_\ini$ and $M_\tar$.
	
	We now prove the if direction.
	It is important to notice that any cycle of length four in $G$ belongs to exactly one gadget and the edge joining two connectors; recall the edges $e_v$ and $e_w$ in \figurename~\ref{fig:edge_configuration}(b).
	Therefore, even in a whole graph $G$, a flip of edges along a cycle of length four can happen only inside of each gadget (and edges joining two connectors). 
	Suppose that there exists a reconfiguration sequence $M_0, M_1, \ldots , M_\ell$ from $M_0=M_\ini$ to $M_\ell = M_\tar$.
	Notice that, by the construction of gadgets, any perfect matching of $G$
	corresponds to a valid NCL configuration such that some NCL edges may take the neutral orientation.
	In addition, $M_\ini$ and $M_\tar$ correspond to valid NCL configurations without any neutral orientation. 
	Pick the first index $i$ in the reconfiguration sequence $M_0, M_1, \ldots , M_\ell$
	which corresponds to changing the direction of an NCL edge $vw$ to the neutral orientation.
	Then, since the neutral orientation contributes to neither $v$ nor $w$, we can simply ignore the change of the NCL edge $vw$ and keep the direction of $vw$ as the same as the previous direction.  
	By repeating this process and deleting redundant orientations if needed, we can obtain a sequence of valid adjacent orientations between $C_\ini$ and $C_\tar$ such that no NCL edge takes the neutral orientation. 
\end{proof}

\subsection{Proof of Corollary~\ref{cor:kfactor}}
\begin{proof}
	We have a simple polynomial-time reduction from {\sc Perfect Matching Reconfiguration}. 
	Let $G, M_\ini, M_\tar$ be an instance of {\sc Perfect Matching Reconfiguration}. 
	We create a new graph $G'$ as follows. 
	First $G'$ contains a copy of $G$, and then we add to it $(k-1) |V(G)|$ new vertices $x_i^j$ with $i \le n$ and $j \le k-1$. 
	We create the edge between $x_i^j$ and $x_i$ for every $i,j$ and create the edges between $x_i^j$ and $x_i^{j'}$ for every $i$ and every $j \ne j'$. 
	Note that every vertex $x_i^j$ has degree exactly $k$. 
	
	Now consider the following $k$-factors $H_\ini$ and $H_\tar$. 
	$H_\ini$ contains all the edges incident to $x_i^j$ for every $i,j$ and the edges of the perfect matching $M_\ini$. 
	$H_\tar$ contains all the edges incident to $x_i^j$ for every $i,j$ and the edges of the perfect matching $M_\tar$. 
	Since all the edges incident to $x_i^j$ have to be in every $k$-factor, there exists a sequence of flip operations transforming $H_\ini$ into $H_\tar$ if and only if there is a reconfiguration sequence transforming $M_\ini$ into $M_\tar$.
\end{proof}

\subsection{Proof Corollary~\ref{cor:kflip} (sketch)}
\begin{proof}
  Consider the gadgets shown in Figure~\ref{fig:gadget}. We replace each orange
  edge of each gadget by a path on $k-3$ edges. Observe that this does not
  alter the number of perfect matchings on each gadget. Using these gadgets, we
  construct from an NCL instance a graph and two perfect matchings as in the
  proof of Theorem~\ref{the:bipartite}.
  
  It is readily verified that any two perfect matchings on each gadget are
  connected by a sequence of flip operations on cycles of length $k$. Note
  that by the selection of the orange edges, each alternating cycle of length $k$
  in the graph that passes through an edge gadget contains an orange edge.
  Therefore, no alternating cycle of length $k$ involves more than one gadget.
  The \PSPACE-completeness of $k$-\textsc{Perfect Matching Reconfiguration}
  follows from the proof of Theorem~\ref{the:bipartite} and the observation
  that each flip involves precisely an orange edge.
\end{proof}

\section{A Proof omitted from Section~\ref{subsec:interval}}
\begin{proof}[Proof of Theorem~\ref{the:stronglyorder}]
	Suppose we are given a strongly orderable graph $G$ together with a corresponding ordering $(v_1,v_2,\ldots,v_n)$ of its vertices. We will argue that every perfect matching of $G$ can be reconfigured in a linear number of steps into any other. We proceed by induction on the number of vertices. We may assume that $G$ is connected and non-empty. In particular, we have $n \geq 2$. Let $M$ be the canonical perfect matching of $G$ with respect to $(v_1,v_2,\ldots,v_n)$. Let $N$ be an arbitrary perfect matching of $G$. There is an edge $v_1v_p$ in $M$, and an edge $v_1v_q$ in $N$. If $p=q$, we delete both vertices from the graph and apply induction. Assume now that $p \neq q$. By choice of a canonical perfect matching, $p < q$. There is an edge $v_pv_r$ in $N$. From the definition of a strong ordering (with $i=1$, $k=p$, $\ell=q$ and $j=r)$, it follows that the edge $v_qv_r$ belongs to the graph. Therefore, from $N$ we can swap the two edges $v_1v_q$ and $v_pv_r$ for $v_1v_p$ and $v_qv_r$. We can then delete the two vertices $v_1$ and $v_q$ and apply induction.
\end{proof}